\documentclass[journal,twocolumn]{IEEEtran}
\usepackage{amssymb}
\usepackage{amsfonts}

\usepackage{savesym}
\usepackage{amsfonts,subfigure,multicol,color,verbatim,
graphicx,cite,epsfig,amssymb,amsmath,cases,bm,algorithm,
algorithmic,xcolor,multirow} 
\usepackage{amsmath}
\usepackage{array}
\usepackage{enumerate}

\savesymbol{iint}

\restoresymbol{TXF}{iint}

\newtheorem{theorem}{\textbf{Theorem}}

\newtheorem{lemma}{\textbf{Lemma}}

\IEEEoverridecommandlockouts

\begin{document}

\markboth{IEEE JOURNAL ON SELECTED AREAS IN COMMUNICATIONS, VOL. 34,
NO. X, Jan. 2016} {Zheng et al: Wireless-powered cooperative communications
\ldots}

\title{Wireless-Powered Cooperative Communications: Power-Splitting Relaying with Energy Accumulation}
\author{\IEEEauthorblockN{Zheng~Zhou, Mugen~Peng,~\IEEEmembership{Senior~Member,~IEEE}, Zhongyuan~Zhao, Wenbo~Wang, and Rick S. Blum,~\IEEEmembership{Fellow,~IEEE}}

\thanks{Manuscript received Apr. 15, 2015; revised Sep. 5, 2015; accepted Dec.
11, 2015. The work of Z. Zhou, and M. Peng was supported in part
by the National Natural Science Foundation of China under Grant 61361166005, the National High Technology Research and Development
Program of China under Grant 2014AA01A701, and the National Basic Research Program of China (973 Program) (Grant No. 2013CB336600). The work of Z. Zhao was supported by National Natural Science Foundation of China (Grant No. 61501045), and the Fundamental Research Funds for the Central Universities. The work of R. S. Blum was supported by the U. S. Army Research Laboratory and the U. S. Army Research Office under grant numbers W911NF-14-1-0245 and W911NF-14-1-0261, and by the National Science Foundation under Grant Numbers CMMI-1400164 and CCF-1442858. The corresponding author: M. Peng.

Zheng~Zhou (e-mail: {\tt nczhouzheng@gmail.com}), Mugen~Peng (e-mail: {\tt pmg@bupt.edu.cn}), Zhongyuan~Zhao (e-mail: {\tt zyzhao@bupt.edu.cn}) and Wenbo~Wang (e-mail: {\tt wbwang@bupt.edu.cn}) are with the Key Laboratory of Universal Wireless Communications for Ministry of Education, Beijing University of Posts and Telecommunications, China. Rick S. Blum (e-mail: {\tt rb0f@lehigh.edu}) is with the Dept. Electrical and Computer Engineering in Lehigh University, Bethlehem, PA, USA.}}

\maketitle
\vspace*{-1.8em}
\begin{abstract}

A harvest-use-store power splitting (PS) relaying strategy with distributed beamforming is proposed for wireless-powered multi-relay cooperative networks in this paper. Different from the conventional battery-free PS relaying strategy, harvested energy is prioritized to power information relaying while the remainder is accumulated and stored for future usage with the help of a battery in the proposed strategy, which supports an efficient utilization of harvested energy. However, PS affects throughput at subsequent time slots due to the battery operations including the charging and discharging. To this end, PS and battery operations are coupled with distributed beamforming. A throughput optimization problem to incorporate these coupled operations is formulated though it is intractable. To address the intractability of the optimization, a layered optimization method is proposed to achieve the optimal joint PS and battery operation design with non-causal channel state information (CSI), in which the PS and the battery operation can be analyzed in a decomposed manner. Then, a general case with causal CSI is considered, where the proposed layered optimization method is extended by utilizing the statistical properties of CSI. To reach a better tradeoff between performance and complexity, a greedy method that requires no information about subsequent time slots is proposed. Simulation results reveal the upper and lower bound on performance of the proposed strategy, which are reached by the layered optimization method with non-causal CSI and the greedy method, respectively. Moreover, the proposed strategy outperforms the conventional PS-based relaying without energy accumulation and time switching-based relaying strategy.

\end{abstract}

\vspace*{-0.5em}
\begin{keywords}

Wireless-powered communication, power splitting, harvest-use-store, channel state information.
\end{keywords}
\vspace*{-1.2em}
\section{Introduction}
\vspace*{-0.5em}
{\color{red}
The fifth generation ($5$G) communication networks are expected to support new emerging services with high network capacity, as well as a reduced delay and energy consumption. To achieve these requirements, the use of super-dense small cell deployments and centralized resource management, i.e. cloud radio access network, is becoming an appealing approach~\cite{2add1,2add2}. However, to fulfill the desired coverage, some wireless nodes in the $5$G network might need to be deployed in places lacking an external power supply. To this end, energy harvesting approaches that scavenge energy from the ambient environment are recognized as a key enabling technology for these self-sustainable nodes~\cite{2add4}. Meanwhile, cooperative relay communication is a promising approach to enlarge coverage and improve spectral efficiency. Therefore, enabling cooperative relay communications via energy harvesting is becoming a popular concept for green communication, which aims at decreasing power usage, while improving the transmission performance.}

A key concern of the energy harvesting enabled cooperative relay communication is the efficient utilization of harvested power, which is not steadily replenished as in traditional grid-aided communication networks. The issue of improving transmission performance via an efficient utilization of harvested power has been widely studied for conventional energy harvesting techniques, where natural resources, such as solar, wind etc. are used as energy sources~\cite{natu1,2add5}. However, the intermittent and unpredictable nature makes these sources difficult to exploit in certain environments. As an alternative to the conventional energy harvesting techniques, radio-frequency~(RF) energy harvesting techniques are believed to fully unleash the potential gains of energy harvesting, in which RF signals transmitted from the source node can be used as energy sources for cooperative nodes. Moreover, it has been illustrated in~\cite{2-1} that wireless-powered cooperative relay communications can be realized within a boundary distance, which is determined by both the available transferred energy and the minimum energy requirements of the harvesting devices. As a result, the wireless-powered cooperative relay communication is formed and can be regarded as a promising solution to highly energy-efficient networking. Note that, due to the dual-purpose of RF signals, i.e., wireless power transfer~(WPT) and wireless information transfer~(WIT), fundamental changes to the designs of green communication networks are entailed.
\vspace*{-0.5em}
\subsection{Related Works}

As concluded in~\cite{htadd2}, energy harvesting receiving architectures and power management models are two essential units to realize wireless-powered cooperative relay communications. The energy harvesting unit is for energy collection, and there are mainly two types of energy harvesting receiving architectures in the literature. In particular, the first type is based on power splitting~(PS) technique, which splits the received RF signal into two different power streams for separate WPT and WIT, and the other is based on time switching technique, where the received signal at one time slot is used for either WPT or WIT~\cite{swipt1}. Note that, the signal received at one time slot is used for both information processing and power transfer in the PS technique, which is not allowed in the time switching technique~\cite{delay}. Therefore, the time switching technique is suboptimal in terms of efficiently using the available signal power~\cite{eff}, and the PS technique is more suitable for applications with critical delay constraints~\cite{delay}. Besides, the power management unit aims at utilizing the harvested power effectively. Most of the research works conducted so far assumes either harvest-use or harvest-store-use model. In the former model, the harvested energy is directly used and energy accumulation is not allowed due to lack of storage units, while in the later model, the harvested energy is first accumulated and stored with the help of a battery and then adaptively utilized in information transmission~\cite{htadd2}.

The incorporation of harvest-use model and PS-based relaying has been studied in~\cite{swipt3,swipt4,swipt5}, where harvested energy through PS is used up at each information relaying. In particular, in~\cite{swipt3}, a harvest-use PS-based relaying strategy was realized in two-way amplify-and-forward relaying systems, and the outage probability and ergodic capacity were analyzed. In~\cite{swipt5}, a multiple relay system was considered, where several randomly located energy harvesting relays help the transmission between a source-destination pair. The proposed strategy was shown to achieve the same diversity gain as the case with conventional self-powered relays. The multiple antenna configurations were considered in~\cite{swipt4}, where a sophisticated relaying strategy that jointly utilizes harvest-use PS and antenna selection was designed and optimized to improve the achievable rate. Besides, the harvest-use time switching-based relaying was studied in multi-tier uplink cellular networks in~\cite{htadd}, where each user transmits only when the harvested energy at one time slot is sufficient for the designed power control scheme, and no user can save the extra harvested energy for the next time slot. Though the harvest-use model is easy to implement, it would perform better if energy accumulation is allowed to store a part of the harvested power for future usage~\cite{swipt5}, which is known as the harvest-store-use model. This model has been incorporated with time switching-based relaying in three-node relaying networks in~\cite{ht3}, where data relaying was realized when sufficient power was collected through time switching technique, and the remaining power was stored for future usage. The outage probability of the proposed strategy was studied. Besides, a similar strategy was discussed for user equipment relay in device-to-device communications in~\cite{ht4}.

Note that, it is still an open problem to allow energy accumulation in PS-based relaying strategy. Furthermore, the harvest-store-use mode is well justified if the battery has perfect efficiency. However, nearly all practical batteries suffer a storage loss to varying extents, ranging from 10\% to 30\%~\cite{sto}. From the perspective of energy efficiency, a new harvest-use-store model has been proposed in~\cite{ht5}, where the harvest energy is prioritized for use in data transmission while its balance/debt is stored in or extracted from the battery, which avoids unnecessary energy loss in storing. This new model has been combined with the conventional energy harvesting techniques in~\cite{ht5}.
\vspace*{-0.5em}
\subsection{Motivations and Contributions}

To realize an efficient utilization of harvested energy and improve spectral efficiency, energy accumulation is realized in PS-based relaying in this paper. In particular, the combining of PS based-relaying and harvest-use-store model to improve throughput performance is considered. Different from the previous works with time switching-based case in~\cite{ht3,ht4}, information transfer and battery operations including the charging and discharging happen at the same time slot in the PS-based case. Moreover, the enhanced PS technique needs to be designed to support the battery operations. To address these challenging issues is the main concern of this paper, and the main contributions can be summarized as follows:

\begin{itemize}
\item To realize an efficient utilization of harvested energy with the help of a battery, a harvest-use-store PS relaying strategy with distributed beamforming is proposed for the wireless-powered multiple-relay scenario. Specifically, each relay obtains both information and power from the received signals transmitted by the source node via power splitting. Subsequently, with the help of a battery, the harvested power is used to amplify-and-forward the information to the destination through distributed beamforming with adaptive power allocation.
\item To reveal a theoretical bound of the proposed strategy, a throughput maximization problem is formulated and solved with an ideal non-causal channel state information~(CSI) assumption. Though the formulated optimization problem is intractable looking, a layered optimization method that derives the optimal solution is developed. In particular, the joint PS and battery operation design in the proposed strategy is decomposed in two layers, such that the original optimization problem is transformed to a dynamic programming problem with a subproblem requiring optimization embedded in it. Then, to address the non-convex embedded subproblem, an alternating-Dinkelbach optimization is proposed to transfer the subproblem to a convex form. Further, the dynamic programming problem can be solved by using backward induction.
\item To study the throughput performance of the proposed harvest-use-store PS relaying strategy in a general scenario with causal CSI, the proposed layered optimization method is extended by utilizing the statistical properties of the CSI via incorporating a finite-state Markov Chain model. Further, a greedy method is proposed, which requires no information about subsequent time slots. It's shown that the adopted greedy method will use up the harvested energy at each transmission. Simulation results reveal that the advantages of the proposed strategy over the conventional battery-free ones depend on the utilization of information about subsequent time slots.
\end{itemize}

The remainder of this paper is outlined as follows. In Section~II, the system model will be presented and the proposed strategy will be described. Following that, the optimal joint PS and battery operation design will be developed with the non-causal CSI assumption in Section~III. Then, the causal CSI case will be analyzed in Section~IV, followed by numerical results in Section~V and conclusions in the final section.

\section{System Model and Protocol Description}

The system under consideration is a wireless-powered cooperative relay network consisting of a source~($S$), a destination~($D$), and a set of $K$ energy harvesting relays~($R_k, k=1,2,\cdots,K$), where each relay is equipped with a battery. No direct link exists between $S$ and $D$. Considering a half-duplex relay model and a total of $T$ equal length time slots. A harvest-use-store PS strategy is proposed and implemented at each time slot, each of which consists of two equal length phases. In the first phase, each relay obtains both information and power from its received RF signal transmitted by $S$ via power splitting. In the second phase, with the help of battery, the harvested energy is used to amplify-and-forward the information to $D$ through distributed beamforming with adaptive power allocation.
\vspace*{-0.5em}
\subsection{Simultaneous Information and Power Transfer via Power Splitting}

In the first phase of time slot $t$, the received RF signal at the ${R_k}$ can be expressed by
\begin{equation}\label{eq:1}
{y_k}\left( t \right) = {h_k}\left( t \right)\sqrt {{p_S}\left( t \right)} {x_S}\left( t \right) + {{\tilde z}^a_k}\left( t \right),
\end{equation}
where $t$ denotes the time index, and $k$ is the index for relays, ${h_k}(t)$ denotes the complex link gain between $S$ and ${R_k}$, ${x_S}(t)$ denotes the signal transmitted from $S$ during time slot $t$ with normalized power $E({\left| {x_S}(t) \right|^2}) = 1$, and ${p_S}(t)$ is the transmit power at $S$ with ${p_S}(t) \le P$. Finally, ${{\tilde z}^a_k}(t)$ is the additive zero mean variance $\sigma _a^2$ white Gaussian noise (AWGN) in the received signal~\cite{swipt3}, thus ${{\tilde z}^a_k}(t)\sim CN\left( {0},\sigma _a^2\right)$.
\begin{figure}[!htb]
\centering\vspace*{-1em}
\includegraphics[width = 3.4 in]{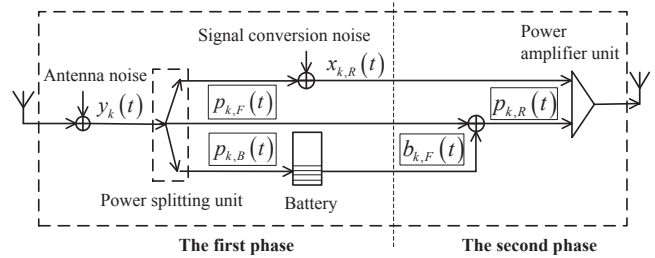}
\caption{Processing at the energy harvesting relay $R_k$ during the $t$-th time slot.}\label{sys2}
\end{figure}

The received RF signal ${y_k}\left( t \right)$ is then divided into three different power streams through PS, which is depicted in Fig.~\ref{sys2}. Specifically, the first part of the received signal, which uses a PS ratio ${\lambda_{k,I}}(t) \in [0,1]$, is down-converted to baseband and will be amplified and transmitted in the next phase. The sampled baseband signal is given by
\begin{equation}\label{eq:2}
x_{k,R}(t) = \sqrt{{\lambda_{k,I}}(t)}\left({h_k}\left( t \right)\sqrt {{p_S}\left( t \right)} {x_S}\left( t \right) + {{z}^a_k}\left( t \right)\right) + {z^b_k}(t),
\end{equation}
where ${z^a_k}(t)\sim CN\left( {0},\sigma _a^2\right)$ is the baseband equivalent noise of the pass band noise ${{\tilde z}^a_k}(t)$, and ${z^b_k}(t)\sim CN\left( {0},\sigma _b^2\right)$ is the sampled AWGN introduced by RF band to baseband signal conversion~\cite{swipt3}. The second part of signal power, which uses a PS ratio $\lambda_{k,F}(t) \in [0,1 - {\lambda_{k,I}}(t)]$, will be used to power the amplify-and-forward process in the next phase. This part of power is described by
\begin{equation}\label{eq:2add}
{p_{k,F}}(t) = \eta_1 {\lambda_{k,F}}(t){\left( {{{\left| {{h_k}\left( t \right)} \right|}^2}{p_S}\left( t \right) + \sigma _a^2} \right)},
\end{equation}
where $\eta_1 \in (0,1]$ is a constant denoting the energy conversion efficiency from signal power to DC power. The last part of signal power, which uses a PS ratio $\lambda_{k,B}(t) = \left(1 - {\lambda_{k,I}}(t) - {\lambda_{k,F}}(t)\right)$, will be used to charge the battery. This part of power is denoted by
\begin{equation}\label{eq:2add2}
{p_{k,B}}'(t) = \eta_1 {\lambda_{k,B}}\left( t \right){\left( {{{\left| {{h_k}\left( t \right)} \right|}^2}{p_S}\left( t \right) + \sigma _a^2} \right)}.
\end{equation}
A finite and discrete battery model is adopted in this paper, where the battery is of size $B_{max} = \alpha P,(\alpha > 0)$ and is discretized into $L+1$ energy levels $\Gamma = \{0,B_{max}/L,\cdots,B_{max}\}$. Note that, this model can closely approximate a continuous battery model when the number of energy levels is sufficiently large~\cite{batt3}. Due to this finite and discrete battery model, the charged power at $R_k$ during time slot $t$ is given by
\begin{equation}\label{eq:11}
\begin{array}{l}
{p_{k,B}}(t) = min\{B_{max} - {B_k}\left( t \right), \frac{{n^*}(t)}{L}B_{max}\},\\
{n^*}(t) = \arg \max \limits_{n(t) \in \{0,\cdots,L\}} \{{n}(t) : \frac{{n}(t)}{L}B_{max} \le \eta_2{p_{k,B}}'(t) \},\\
\end{array}
\end{equation}
where $B_k(t)\in \Gamma$ denotes the energy level at the beginning of time slot $t$, $\eta_2 \in (0,1]$ is the storage efficiency describing the power loss in battery charging~\cite{htadd2}, the first equation is due to the finite property of the battery, and the second one is due to the discrete property of the battery.
\vspace*{-0.5em}
\subsection{Distributed Beamforming with Adaptive Power Allocation}

In the second phase, the destination $D$ receives a signal transmitted from all relays as
\begin{equation}\label{eq:5}
y\left( t \right) = \sum\limits_{k = 1}^K {\beta _k}\left( t \right){{g_k}\left( t \right)e^{j{\theta_k}\left( t \right)}x_{k,R}(t)}  + z\left( t \right),
\end{equation}
where ${g_k}(t)$ denotes the link gain for ${R_k}$ to $D$ and ${z}(t) \sim CN\left( {0},\sigma _D^2\right)$ is the AWGN at $D$, $e^{j{\theta_k}\left( t \right)}$ is derived from the distributed beamforming design~\cite{dis}, where
\begin{equation}\label{eq:8add}
{\theta_k}(t) = -{\left( {\arg {h_k}\left( t \right) + \arg {g_k}\left( t \right)} \right)},
\end{equation}
which cancels the phases of the two links between $R_k$ and $S,D$ respectively, ${\beta_k}(t)$ is the amplification gain depicted by
\begin{equation}\label{eq:6}
{\beta_k}(t) = \sqrt {\frac{p_{k,R}(t)}{{{\lambda_{k,I}}(t) {\left( {{{\left| {{h_k}\left( t \right)} \right|}^2}{p_S}\left( t \right) + \sigma _a^2} \right)}+\sigma ^2_b}}},
\end{equation}
where ${p_{k,R}}(t)$ is the transmit power at $R_k$, this power is composed of two different parts:

\begin{equation}\label{eq:7}
{p_{k,R}}(t) = {p_{k,F}}(t) + {b_{k,F}}\left( t \right),
\end{equation}
where $p_{k,F}$ is provided by the PS operation in~\eqref{eq:2add}, and ${b_{k,F}}(t) \in \Gamma$ is provided by the battery discharging. As a result, the energy level at $R_k$ at the end of time slot $t$~(which is also the beginning of time slot $(t+1)$) is
\begin{equation}\label{eq:8}
{B_k}(t+1) = {B_k}\left( t \right) + {p_{k,B}}(t) - {b_{k,F}}\left( t \right).
\end{equation}

Note that, any power consumption at the relays for purposes other than for transmission is assumed negligible~\cite{ht3,swipt3}. Due to the fact that the antenna noise ${z^a_k}\left( t \right)$ has a negligible impact on both the information processing and energy harvesting~\cite{noi}, it is ignored in the following analysis by setting $\sigma _a^2 = 0$. Substituting~\eqref{eq:1} and~\eqref{eq:2} into~\eqref{eq:5}, the received signal-to-noise-ratio (SNR) at $D$ at time slot $t$ can be expressed as
\begin{equation}\label{eq:9}
{SNR}(t) = \frac{{{{{p_S}\left( t \right)\left( {\sum\limits_{k = 1}^K { {{\beta_k}(t){{\left| {{h_k}\left( t \right){g_k}\left( t \right)} \right|}}\sqrt{\lambda _{k,I}\left( t \right)}} } } \right)}^2}}}{{\sum\limits_{k = 1}^K {{{\beta_k}(t)}^2{{\left| {{g_k}\left( t \right)} \right|}^2}\sigma _b^2}  + \sigma _D^2}}.
\end{equation}

The overall system throughput after $T$ transmissions is described by
\begin{equation}\label{eq:10}
R_{total} = \frac{1}{2}\sum\limits_{t = 1}^T {\log \left( {1 + SNR\left( t \right)} \right)},
\end{equation}
where the factor $1/2$ is due to the half-duplex relaying mode.

According to~\eqref{eq:6},~\eqref{eq:7},~\eqref{eq:9}, and~\eqref{eq:10}, four key factors should be jointly considered to optimize the overall throughput, including information transfer~(i.e., designing $\lambda_{k,I}(t),\forall k, \forall t$), power transfer~(i.e., designing $\lambda_{k,F}(t),\forall k, \forall t$), battery charging~(i.e., designing $\lambda_{k,B}(t),\forall k, \forall t$), and battery discharging~(i.e., designing $b_{k,F}(t),\forall k, \forall t$). The first three factors are mainly determined by the PS, while the last one is only determined by the battery operation. Moreover, the harvested energy~(i.e., ${p_{k,B}}(t),\forall k, \forall t$) obtained from the PS affects the throughput during the subsequent time slots via battery charging and discharging. To this end, the PS and battery operations are coupled with distributed beamforming, which makes the throughput maximization problem intractable.

\section{Optimal Joint Power Splitting and Battery Operation Design}

In this section, the proposed harvest-use-store PS relaying strategy is optimized toward the goal of throughput maximization with the CSI of $T$ time slots known before transmission, which is commonly known as non-causal CSI assumption. To solve the resulting intractable looking problem, the throughput maximization problem is described by a dynamic programming problem equivalently. However, such a solution is computationally prohibitive to implement since the number of possible joint PS and battery operations to be evaluated is infinite. Therefore, the joint PS and battery operations are decomposed equivalently, and the resulting embedded problem and overall problem are formulated. Moreover, the optimal solutions to the embedded and overall problems can be approached through iterative algorithms and backward induction. In this way, the optimal solution toward throughput maximization is approached.
\vspace*{-0.5em}
\subsection{Decoupling of the Throughput Maximization Problem}

We first simplify the formulation of the throughput maximization problem by utilizing insights about the optimization variables. According to~\eqref{eq:9} and~\eqref{eq:10}, the throughput $R_{total}$ is a strictly increasing function of $p_S(t),\forall t$, when other variables are fixed. Since $p_S(t)\leq P $, the optimum choice is
\begin{equation}\label{eq:12add}
p_S(t) = P. \; \forall t
\end{equation}

Next, the optimized $\lambda_{k,B}(t)$ in~\eqref{eq:2add2} and~\eqref{eq:11} should satisfy
\begin{equation}\label{eq:11add}
{p_{k,B}}'(t) = {p_{k,B}}(t) = \frac{{n^*}(t)}{L}B_{max}, \; \forall t,\forall k
\end{equation}

Note that the throughput $R_{total}$ can generally be further improved by reassigning the PS ratios. $\lambda_{k,B}(t)$ is decreased to meet~\eqref{eq:11add}, $\lambda_{k,F}(t)$ remains the same, and $\lambda_{k,I}(t)$ is increased, which can be seen according to~\eqref{eq:2add},~\eqref{eq:6},~\eqref{eq:7}, and~\eqref{eq:9}. The result in~\eqref{eq:11add} implies that no split signal power for battery charging would be discarded due to the finite and discrete property of the battery, because the discarded part could have been assigned for information processing and improve the transmission performance via the more proper PS ratios.

To make the analysis more concise, the energy level variation is introduced to describe the effects of battery charging and discharging, which is denoted by ${v_k}\left( t \right)$ at $R_k$ during time slot $t$, such that
\begin{equation}\label{eq:12}
\begin{array}{l}
{v_k}\left( t \right) = {B_k}(t) - {B_k}(t + 1)
  \\\mathop
= \limits^{\left( a \right)} {b_{k,F}}\left( t \right) - \eta_1 \eta_2 \lambda_{k,B}\left( t \right)P{\left| {{h_k}\left( t \right)} \right|^2},
\end{array}
\end{equation}
where the equation $(a)$ is derived from~\eqref{eq:2add2},~\eqref{eq:8} and~\eqref{eq:11add}. To avoid unnecessary power loss caused by battery charging~(due to $\eta_2$), it's assumed that at each battery, the charging and discharging operations would not occur during the same time slot. In particular, if the energy level at $R_k$ decreases at time slot $t$~(i.e., $v_k(t) \ge 0$), the battery is not charged in the first phase~(i.e., ${\lambda_{k,B}}\left( t \right) = 0$), and discharged in the second one, which implies that the stored battery power is utilized to back up the transmit power. On the other hand, if the energy level increases~(i.e., $v_k(t) \le 0$), the battery is charged in the first phase, and not discharged in the second one~(i.e., ${b _{k,F}}\left( t \right) = 0$), which implies that a part of the harvested power is stored for future usage. Based on~\eqref{eq:12}, these relationships can be described by

\begin{equation}\label{eq:14}
\begin{array}{l}
{b_{k,F}}\left( t \right) = \max \left\{ {0,{{v_k}\left( t \right)}} \right\},\\
\lambda_{k,B}(t) =  - \min \left\{ {0,\frac{{{v_k}\left( t \right)}}{{\eta_1 \eta_2 P{{\left| {h_k\left( t \right)} \right|}^2}}}} \right\}.
\end{array}
\end{equation}

Substituting~\eqref{eq:2add},~\eqref{eq:12add} and~\eqref{eq:14} into~\eqref{eq:7}, the transmit power for $R_k$ at time slot $t$ is rewritten as
\begin{equation}\label{eq:14add}
\begin{array}{l}
{p_{k,R}}\left( t \right) = {\eta _1}\left(1 - \lambda_{k,I}\left( t \right)\right) P{{\left| {{h_k}\left( t \right)} \right|}^2} \\
\quad \quad \quad \quad \quad \quad + \min \left\{ {0,\frac{{{v_k}\left( t \right)}}{{{\eta _2}}}} \right\} + \max \left\{ {0,{v_k}\left( t \right)} \right\}.
\end{array}
\end{equation}

As a result, the throughput maximization problem can be formulated as
\begin{equation}\label{eq:15}
\begin{array}{l}
\left( {P1} \right):\mathop {\max }\limits_{{{\bold{I}}(t),\bold{v}}(t),\forall t} {R_{total}}\\
s.t. C1: 0 \le {\lambda _{k,I}}\left( t \right) \le 1 + \min \left\{ {0,\frac{{{v_k}\left( t \right)}}{{{\eta _1}{\eta _2}P{{\left| {{h_k}\left( t \right)} \right|}^2}}}} \right\},\;\forall k,\forall t\\
C2: {v_k}\left( t \right) = n(t)\frac{{{B_{max}}}}{L},n(t) \in \left\{ { - L, \cdots ,0, \cdots ,L} \right\},\forall k,\forall t \\
C3: {B_k}\left( 1 \right) - \sum\limits_{i = 1}^{t - 1} {{v_k}\left( i \right)} - B_{max} \le {v_k}\left( t \right) \\
\quad \quad \quad \quad \quad \quad \le {B_k}\left( 1 \right) - \sum\limits_{i = 1}^{t - 1} {{v_k}\left( i \right)},\;\forall k,\forall t
\end{array}
\end{equation}
where ${\bold{I}}(t)=(\lambda_{1,I}(t),\cdots,\lambda_{K,I}(t))$ describes the information transfer design at time slot $t$ and ${\bold{v}}(t)=({v_1}\left( {t} \right),\cdots,{v_K}\left( {t} \right))$ denotes the energy level variation design. The first constraint $C1$ is derived from the definitions of PS rations and the constraint in~\eqref{eq:14}, which implies that the range of allowable information transfer designs changes with the energy level variation designs. The second constraint $C2$ is derived from the definition of $B_k(t)$ and $v_k(t)$, which is due to the finite and discrete property of the battery. The third constraint $C3$ is derived from~\eqref{eq:12}, which implies that design of the energy level variation is coupled over time. Since allowable values for ${v_k}(t),\forall k,\forall t$ are discrete, Problem (P1) is a mixed integer optimization problem and non-convex.

To handle this time coupled joint PS and battery operation design in Problem (P1), the following components are defined to transform Problem (P1) into a dynamic programming problem equivalently.
\begin{enumerate}[1)]
\item battery state: The battery state describes the energy levels at all batteries at the beginning of each transmission, which is described by $\bold{S}(t) = \{{B_1}\left({t} \right),\cdots,{B_K}\left( {t} \right)\}$ at time slot $t$. Without loss of generality, the initial battery state is set as $\bold{S}(1) = \{0,\cdots,0\}$.
\item decision: The decision denotes the energy level variation designs and the information transfer designs at all relays, which is denoted by~$\bold{D}(t) = \left({\bold{v}}(t), {\bold{I}}(t)\right)$ at time slot $t$.
\item state evolution: The state evolution describes the change of battery states over two adjacent time slots, which can be denoted by
    \begin{equation}\label{eq:add3}
    \bold{S}(t) - \bold{S}(t+1) = \bold{v}\left( {t} \right),
    \end{equation}
    at time slot $t$.~\eqref{eq:add3} can be interpreted as ${B_k}\left({t} \right) - {B_{k}}\left({t+1} \right) = v_k(t), \forall k$, which is derived from~\eqref{eq:12}.
\end{enumerate}

Note that, the energy level variation design results in the state evolution. Moreover, the range of allowable energy level variation designs is determined by battery states
\begin{equation}\label{eq:15add}
{B_k}\left( t \right)  - B_{max} \le {v_k}\left( t \right) \le {B_k}\left( t \right),\;\forall k
\end{equation}
which is derived from~\eqref{eq:12} and the constraint $C3$ in Problem (P$1$).

\begin{enumerate}
\setcounter{enumi}{3}
\item payoff: The payoff is the throughput at each time slot. According to~\eqref{eq:9}, the payoff is determined by the battery state ${\bold{S}}(t)$, the decision ${\bold{D}}(t)$ and the time slot $t$, which can be described by
    \begin{equation}\label{eq:31}
    R({\bold{S}}(t),\bold{D}(t),t) =  \frac{1}{2} { \log \left( {1 + SNR\left( t \right)} \right)}.
    \end{equation}
\item summation of payoffs: The summation of payoffs starting from time slot $T$ and summing backward to the first time slot~(to facilitate backward induction) is defined as
  \begin{equation}\label{eq:32}
  \begin{array}{l}
  U({\bold{S}}(t),\bold{D}(t),t) = R({\bold{S}}(t),\bold{D}(t),t) \\
  \quad \quad \quad + U({\bf{S}}(t+1),{{\bf{{D}}}}(t+1),(t+1)), \; \forall (t \ne T)\\
  U({\bf{S}}(T),{\bf{D}}(T),T) = {R}({\bf{S}}(T),{\bf{D}}(T),T).
  \end{array}
  \end{equation}

\end{enumerate}

In this way, the throughput in Problem~(P1) is the summation of payoffs through $T$ time slots~(i.e., $R_{total} = U\left({\bf{S}}\left( 1 \right), {{\bf{{D}}}}(1) , 1 \right)$).
Based on equation~\eqref{eq:32} and Bellman equation~\cite{dyn}, the optimality equations to solve Problem (P1) can be written as
\begin{equation}\label{eq:addop1}
  \begin{array}{l}
  U^*({\bold{S}}(t),t) = \mathop{\max}\limits_{{\bf{v}}(t), {\bf{I}}(t)}  {R}({\bold{S}}(t),{\bf{v}}(t), {\bf{I}}(t),t) \\
  \quad \quad \quad + U^*({\bf{S}}(t+1),(t+1)), \; \forall t \ne T\\
  s.t. C1, C2, C3.
  \end{array}
\end{equation}
\begin{equation}\label{eq:addop2}
  \begin{array}{l}
  U^*({\bf{S}}(T),T) = \mathop{\max}\limits_{{\bf{v}}(T), {\bf{I}}(T)} {R}({\bf{S}}(T),{\bf{v}}(T),{\bf{I}}(T),T),\\
  s.t. C1, C2, C3,
  \end{array}
\end{equation}
as in many typical dynamic programming problems, we can work on one time slot at a time, starting from the last one~(i.e., $U^*({\bf{S}}(T),T)$) and working backward~(i.e., until $U^*({\bf{S}}(1),1)$), using what is called backward induction~\cite{dyn}. As a result, the optimal set of decisions can be found by the backward induction algorithm which essentially evaluates all possible sets of decisions and eventually picks the best\footnote{The Viterbi algorithm can be viewed as dynamic programming}.
\vspace*{-0.5em}
\subsection{Decomposing of the Joint PS and Battery Operation}

Though the optimal solution to Problem (P1) can be obtained through dynamic programming, such a solution is computationally prohibitive to implement because the PS ratios ${\bf{I}}(t)$ in each decision are continuous variables, which makes the number of possible decisions infinite. To overcome this computationally prohibitive issue, all possible decisions are evaluated in a decomposed manner, and thus the optimization problems in~\eqref{eq:addop1} and~\eqref{eq:addop2} are solved through a two step procedure. {\color{red}The basic idea is that we optimize $R({\bold{S}}(t),\bold{v}(t),\bold{I}(t),t)$ over $I(t)$ for each of the possible values of $v(t)$ in (23) and (24), where these optimized results are denoted by $\bold{I}^*({\bold{v}}(t))$. Since we have the optimum $\bold{I}^*({\bold{v}}(t))$ for each of the possible values of ${\bold{v}}(t)$, we can use this to eliminate $\bold{I}(t)$ in (23) and (24).} Correspondingly, the formulated subproblem in {\color{red}the first step} is given by
\begin{equation}\label{eq:add33}
\begin{array}{l}
  \mathop{\max}\limits_{{\bf{I}}(t)} R({\bold{S}}(t),\bold{v}(t),\bold{I}(t),t) \\
  = \mathop {\max }\limits_{\lambda _{1,I}\left( t \right),\cdots,\lambda _{K,I}\left( t \right)} \; \frac{1}{2} { \log \left( {1 + SNR\left( t \right)} \right)},\\
s.t. \quad C1.
\end{array}
\end{equation}
which is derived from~\eqref{eq:addop1},~\eqref{eq:addop2} and~\eqref{eq:31}, and is called the embedded problem in the following. Note that, the embedded problem needs to be solved for each possible $\bold{v}(t)$ with a certain $({\bold{S}}(t),t)$.

In the second step, by substituting the optimal solutions $\bold{I}^*({\bold{v}}(t))$ to~\eqref{eq:addop1} and~\eqref{eq:addop2}, this dynamic programming problem can be simplified to
\begin{equation}\label{eq:33}
\begin{array}{l}
  {U^*}\left( {{\bf{S}}\left( t \right)},t \right) = \mathop {\max}\limits_{{\bf{v}}(t)} \left\{ R({\bold{S}}(t),\bold{v}(t),\bold{I}^*({\bf{v}}(t)),t) \right.\\
  \quad \quad \quad \quad \quad \quad \left.  + U^*\left( {{\bf{S}}\left( {t + 1} \right)},(t+1) \right) \right\}, \forall (t \ne T)\\
  s.t. \quad C2, C3.
\end{array}
\end{equation}

\begin{equation}\label{eq:addb2}
\begin{array}{l}
  {U^*}\left( {{\bf{S}}\left( T \right)},T \right) = \mathop {\max}\limits_{{\bf{v}}(t)} {R}({\bf{S}}(T),{\bf{v}}(T),{\bold{I}}^*({\bf{v}}(T)),T),\\
  s.t. \quad C2, C3.
  \end{array}
\end{equation}
Note that, since each possible ${\bold{v}}(t)$ are evaluated with the corresponding $\bold{I}^*({\bold{v}}(t))$, the number of decisions to be evaluated is $(L+1)^K$ for each state ${\bf{S}}(t)$ at each time slot $t$. To find the optimal set of decisions in~\eqref{eq:33} and~\eqref{eq:addb2} is called the overall problem in the following. In this way, the joint PS and battery operation is decomposed, and Problem (P1) is transformed to a dynamic programming problem~(overall problem) and an subproblem requiring optimization~(embedded problem).
\vspace*{-0.5em}
\subsection{Optimizing of the Embedded Problem and the Overall Problem}

We first solve the embedded problem in~\eqref{eq:add33} with each allowable battery operation ${\bold{v}}(t)$. Since the objective function in~\eqref{eq:add33} is monotonically increasing in terms of $SNR(t)$, this payoff maximization problem is clearly equivalent to the SNR maximization problem~\cite{dis}, which is described by
\begin{equation}\label{eq:add333}
\begin{array}{l}
\mathop {\max }\limits_{\lambda _{1,I}\left( t \right),\cdots,\lambda _{K,I}\left( t \right)} \; SNR\left( t \right) \\
s.t. \quad 0 \le \lambda_{k,I}\left( t \right) \le 1 + \min \left\{ {0,\frac{{{v_k}\left( t \right)}}{{\eta_1 \eta_2 P{{\left| {h_k\left( t \right)} \right|}^2}}}} \right\}. \; \forall \; k
\end{array}
\end{equation}

By substituting~\eqref{eq:6}~\eqref{eq:9} and~\eqref{eq:14add} in~\eqref{eq:add333} and omitting $\lambda_{k,I}(t),h_k(t),g_k(t)$ and $v_k(t)$'s dependence on $t$ since $t$ is constant during the optimization, this SNR maximization problem can be rewritten as
\begin{equation}\label{eq:43}
\begin{array}{l}
(P2): \mathop {\max }\limits_{ {x_1, \cdots ,x_K} } J({x_1, \cdots ,x_K}) \\
\quad \quad \quad\quad \quad \quad= \frac{{{{\left( {\sum\limits_{k = 1}^K {\sqrt {\eta_1 {{\left| {{g_k}} \right|}^2}\left( {{a_k} - {x_k}} \right)\left( {1 - \frac{{\sigma _b^2}}{{{x_k}}}} \right)} } } \right)}^2}}}{{\sum\limits_{k = 1}^K {\eta_1 {{\left| {{g_k}} \right|}^2}\left( {{a_k} - {x_k}} \right)\frac{{\sigma _b^2}}{{{x_k}}}}  + \sigma _D^2}},\\
s.t. \sigma _b^2 \le {x_k} \le \left( {1 + \min \left\{ {0,\frac{{{v_k}}}{{\eta_1 \eta_2 P{{\left| {{h_k}} \right|}^2}}}} \right\}} \right)P|{h_k}{|^2} + \sigma _b^2, \forall k
\end{array}
\end{equation}
where ${a_k} =  {P|{h_k}{|^2} + \max \left\{ {0,\frac{{{v_k}}}{{{\eta _1}}}} \right\} + \min \left\{ {0,\frac{{{v_k}}}{{{\eta _1}{\eta _2}}}} \right\}} + \sigma _b^2$, and ${x_k} = \lambda_{k,I} P|{h_k}{|^2} + \sigma _b^2$.

A key challenge in solving Problem (P$2$) is the lack of convexity in the problem formulation, thus alternating optimization~\cite{alt} and Dinkelbach algorithm~\cite{fractional} are utilized jointly to transfer Problem (P$2$) into a convex form, which is called alternating-Dinkelbach for short in this paper. First, the $K$ variables~(i.e., $x_1,\cdots,x_K$) in Problem (P$2$) are handled by alternating optimization, where one variable is updated at a time while fixing the others. Specifically, the decomposed subproblem to optimize $x_j$ can be given by
\begin{equation}\label{eq:46}
\begin{array}{l}
(P3):\mathop {\max }\limits_{{x_j}} \quad \mathop \frac{F_1(x_j)}{F_2(x_j)}\\
s.t. \; \sigma _b^2 \le {x_j} \le \left( {1 + \min \left\{ {0,\frac{{{v_j}}}{{{\eta _1}{\eta _2}P{{\left| {{h_j}} \right|}^2}}}} \right\}} \right)P|{h_j}{|^2} + \sigma _b^2,
\end{array}
\end{equation}
where $F_1({x_j}) = \left( {\sqrt {{\eta _1}{{\left| {{g_j}} \right|}^2}\left( {{a_j} - {x_j}} \right)\left( {1 - \frac{{\sigma _b^2}}{{{x_j}}}} \right)} } \right.\\
{\left. { + \sum\limits_{k \ne j} {\sqrt {{\eta _1}{{\left| {{g_k}} \right|}^2}\left( {{a_k} - {x_k}} \right)\left( {1 - \frac{{\sigma _b^2}}{{{x_k}}}} \right)} } } \right)^2} \ge 0$, $F_2(x_j) = {{{\eta _1}{{\left| {{g_j}} \right|}^2}\left( {{a_j} - {x_j}} \right)\frac{{\sigma _b^2}}{{{x_j}}} + \sum\limits_{k \ne j} {{\eta _1}{{\left| {{g_k}} \right|}^2}\left( {{a_k} - {x_k}} \right)\frac{{\sigma _b^2}}{{{x_k}}}}  + \sigma _D^2}} >0$. Then the objective function in Problem (P$3$) is transferred from the fractional form to the subtractive form via Dinkelbach algorithm~\cite{fractional}. In particular, defining a parameter $q$ as
\begin{equation}\label{eq:add43}
q = \mathop \frac{F_1(x_j)}{F_2(x_j)},
\end{equation}
we formulate a subtractive form optimization problem with a given parameter $q$ as
\begin{equation}\label{eq:45}
\begin{array}{l}
(P4): \; F(q) = \mathop {\max }\limits_{{x_j}} \; F_1(x_j) - q F_2(x_j)\\
s.t.\quad \sigma _b^2 \le {x_k} \le \left( {1 + \min \left\{ {0,\frac{{{v_k}}}{{\eta_1 \eta_2 P{{\left| {{h_k}} \right|}^2}}}} \right\}} \right)P|{h_k}{|^2} + \sigma _b^2.
\end{array}
\end{equation}

To solve Problem (P$3$), we first present the following lemma.
\begin{lemma}
\begin{equation}\label{eq:44addd}
q' = \mathop \frac{F_1(x_j')}{F_2(x_j')} = \mathop {\max }\limits_{{x_j}} \; \mathop \frac{F_1(x_j)}{F_2(x_j)}
\end{equation}
if, and only if,
\begin{equation}\label{eq:44}
\begin{array}{l}
{F(q')} = \mathop {\max }\limits_{{x_j}} \;
{F_1(x_j)} - {q'}{F_2(x_j)}\\
= F_1\left( {x_j'} \right) - {q'}F_2\left( {x_j'} \right) = 0.
\end{array}
\end{equation}
\end{lemma}

\begin{proof}
\textbf{Lemma $1$} can be proved by following a similar approach as in~\cite{fractional}.
\end{proof}
\textbf{Lemma $1$} reveals that to solve Problem (P$3$) with an objective function in fractional form, there exists a corresponding Problem (P$4$) in subtractive form. Moreover, \textbf{Lemma $1$} provides the condition about when the two problem formulations lead to the same optimal solution~$x_j'$. To reach the condition in~\eqref{eq:44}, we focus on Problem (P$4$) first.

\begin{lemma}
The optimization objective in Problem (P$4$) is a concave function in terms of $x_j$.
\end{lemma}

\begin{proof}
Please refer to Appendix A.
\end{proof}

According to \textbf{Lemma $2$ }and the fact that the feasible set for $x_j$ is convex, which can be seen in~\eqref{eq:45}, Problem (P$4$) is a convex optimization problem. However, the complicated expression of $F(q)$ makes it difficult to derive a closed-form solution through Karush-Kuhn-Tucker~(KKT) conditions, thus a numerical method~(i.e., bisection method~\cite{bisection}) is adopted here to handle Problem (P$4$).

Then, we follow a similar approach as in~\cite{fractional}~(known as the Dinkelbach algorithm) to derive the solution that satisfies~\eqref{eq:44}, which is summarized in \textbf{Algorithm $1$}.
\begin{algorithm}
\caption{Information Transfer Design Optimization at Relay $R_j$} \label{alg:Framwork}
\begin{algorithmic}[1]
\STATE \textbf{Input:} Fixed values of $x_k,\forall (k \ne j)$.
\STATE \textbf{Initialization:}
\STATE ~Set the parameter as $q = 0$, the index of iteration as $n_1 = 0$, and the judgement of convergence as $conv_1 = 0$,
\STATE ~Set the maximum number of iterations as $n_{max}^1$, and the threshold of termination as $\Delta_1$, which is a constant that approaches $0$.
\STATE \textbf{Repeat:}
\STATE ~Set ${n_1} = n_1 + 1$,
\STATE ~Solve Problem (P$4$) through bisection method~\cite{bisection}, and mark the optimal solution as $x_j^{n_1}$,
\STATE ~~\textbf{If} ${F_1(x_j^{n_1})} - {q}{F_2(x_j^{n_1})} < \Delta_1$
\STATE ~~~Mark the optimal solution as $x_j' = x_j^{n_1}$, and set $conv_1 = 1$.
\STATE ~~\textbf{else}
\STATE ~~~Update $q$ according to~\eqref{eq:add43}.
  \STATE \textbf{Until:} $conv_1 = 1$ or $n_1 = n_{max}^1$.
\STATE \textbf{Return:} The optimized information transfer design for $R_j$ as $x_j'$.
\end{algorithmic}
\end{algorithm}

\begin{lemma}
\textbf{Algorithm $1$} converges to $x_j'$ and $q'$, which satisfy~\eqref{eq:44} in \textbf{Lemma $1$}.
\end{lemma}
\begin{proof}
For the purpose of explanation, at the $n$-th iteration, denote the parameter as $q_n$, and the optimized solution to Problem (P$4$) as ${x_j^n}$. According to~\eqref{eq:add43}, the parameter is updated by $q_{n+1} = \frac{F_1(x_j^n)}{F_2(x_j^n)}$ in the next iteration. Assume that the iteration process will not be terminated in the $(n+1)$-th iteration.

First, it's shown that, the optimized value $F(q)$ in Problem (P$4$) is non-negative.
\begin{equation}\label{eq:b2}
\begin{array}{l}
F\left( q_{n+1} \right)
 = \mathop {\max }\limits_{x_j} \; {F_1\left( {x_j}\right) - {q_{n+1}}{F_2}\left( {x_j}\right)}\\
 \ge F_1\left( {x_j^{n}} \right) - {q_{n+1}}F_2\left( {x_j^{n}} \right) = 0.
\end{array}
\end{equation}

Next, it's revealed that the parameter q increases after each iteration. According to \textbf{Lemma $1$} and~\eqref{eq:b2}, since the iteration process is not terminated in the $(n+1)$-th iteration, $F\left( q_{n} \right) > 0$ and $F\left( q_{n+1} \right) > 0$ must be true. Thus, $F\left( q_{n} \right)$ can be expressed as
\begin{equation}\label{eq:b3}
\begin{array}{l}
F\left( q_{n} \right)= F_1\left( {x_j^{n}} \right) - {q_{n}}F_2\left( {x_j^{n}} \right)\\
 = \left( {{q_{n + 1}} - {q_n}} \right)F_2\left( {x_j^{n}} \right) > 0,
\end{array}
\end{equation}
since $F_2\left( {x_j^{n}} \right) > 0$, ${q_{n + 1}} > {q_n}$ must be true.

Further, it's shown that $F\left( q \right)$ decreases after each iteration. Due to the increasing of $q$ after each iteration, $F\left( q_{n} \right)$ can be expressed as
\begin{equation}\label{eq:b1}
\begin{array}{l}
F\left( q_{n} \right)
= \mathop {\max }\limits_{x_j} \; {F_1\left( {x_j}\right) - {q_{n}}{F_2}\left( {x_j}\right)}\\
 \ge F_1\left( {x_j^{n+1}} \right) - {q_{n}}F_2\left( {x_j^{n+1}} \right)\\
 > F_1\left( {x_j^{n+1}} \right) - {q_{n+1}}F_2\left( {x_j^{n+1}} \right)\\
  = F\left( q_{n+1} \right).
\end{array}
\end{equation}

Based on the properties in~\eqref{eq:b2} and~\eqref{eq:b1}, after a large enough number of iterations, $F\left( q \right)$ finally approaches to $0$. As a result, \textbf{Algorithm $1$} converges to $x_j'$ and $q'$, which satisfy~\eqref{eq:44} in \textbf{Lemma $1$} and $x_j'$ is the optimal solution to Problem (P$3$).
\end{proof}

Note that, the decomposed subproblems to optimize $x_j,j\in \{1,\cdots,K\}$ have the same problem formulation as Problem (P$3$) except that the index number $j$ is different, thus they can be solved using \textbf{Algorithm $1$}. As a result, Problem (P$2$) is solved using the proposed alternating-Dinkelbach optimization, and the corresponding procedure is summarized in \textbf{Algorithm $2$}.
\begin{algorithm}
\caption{Embedded Problem Optimization} \label{alg:Framwork}
\begin{algorithmic}[1]
\STATE \textbf{Input:} A battery state $\bold{S}(t)$ and a energy level variation design $\bold{v}$.
\STATE \textbf{Initialization:}
\STATE ~Set the index of iteration as $n_2 = 0$, and the judgement of convergence as $conv_2 = 0$,
\STATE ~Set the maximum number of iterations as $n_{max}^2$, and the threshold of termination as $\Delta_2$, which is a constant that approaches $0$,
\STATE ~Set the optimized objective value of Problem (P$2$) as $J^* = 0$.
\STATE \textbf{Repeat:}
\STATE ~Set $n_2 = n_2 + 1$, calculate the index of the decomposed subproblem as $j = {n_2} - \left\lfloor {\frac{{{n_2}}}{K}} \right\rfloor *K$, where $\left\lfloor \right\rfloor$ will round down.
\STATE ~Derive $x_j'$ using \textbf{Algorithm $1$} with the fixed values of $x_{k}(t),\forall (k \ne j)$, and calculate $J(x_1,\cdots,x_j',\cdots,x_K)$ in~\eqref{eq:43}.
\STATE ~\textbf{If} $J(x_1,\cdots,x_j',\cdots,x_K) - J^* < \Delta$
\STATE ~~Set $conv_2 = 1$, and mark the optimal solution as $[x_1^*,\cdots,x_K^*]$,
\STATE ~~Calculate $\lambda_{k,I}^*,\forall k$ according to the definition ${x_k} = \lambda_{k,I} P|{h_k}{|^2} + \sigma _b^2$.
\STATE ~\textbf{else}
\STATE ~~Update $n_2 = n_2 + 1$, and set $J^* = J(x_1,\cdots,x_j',\cdots,x_K)$.
\STATE \textbf{Until:} $conv_2 = 1$ or $n_2 = n_{max}^2$.
\STATE \textbf{Return:} The optimal embedded information transfer design $\bold{I}^* = (\lambda_{1,I}^*,\cdots,\lambda_{K,I}^*)$.
\end{algorithmic}
\end{algorithm}

\begin{lemma}
The proposed \textbf{Algorithm $2$} converges to the optimal solution to Problem (P$2$).
\end{lemma}

\begin{proof}
Please refer to Appendix B.
\end{proof}

When the optimal solution to the embedded problem is obtained through Algorithm $2$, the overall problem can be solve. In particular, the stored power at batteries should be used up at the last time slot or they will be wasted, which is described by ${{\bf{{v}}}}^*(T) = {\bf{S}}(T)$. Thus, ~\eqref{eq:addb2} can be solved as
\begin{equation}\label{eq:34}
{U^*}\left( {{\bf{S}}\left( T \right)},T \right) = {R}({\bf{S}}(T),{\bf{S}}(T),{\bold{I}}^*({\bf{S}}(T)),T).
\end{equation}
Then, optimality equations in~\eqref{eq:33} can be solved through backward induction algorithm~\cite{dyn}. In conclusion, the corresponding procedure to solve Problem (P$1$) is described in \textbf{Algorithm $3$}. Note that, though the performance of the optimal joint PS and battery operation design derived from \textbf{Algorithm $3$} cannot generally be achieved in practice because it requires non-causal CSI, a theoretical upper bound on performance of the proposed harvest-use-store PS relaying strategy is provided.
\begin{algorithm}
\caption{Joint Power Splitting and Battery Operation Optimization with Non-causal CSI} \label{alg:Framwork}
\begin{algorithmic}[1]
\STATE \textbf{Input:} Non-causal channel state information~(i.e., $h_k(t),g_k(t),\forall k, \forall t$).
\STATE \textbf{Initialization:}
\STATE ~Set the time index as $t=T$;\\
\STATE ~\textbf{for all} allowable battery state~${\bf{S}}(T)$\\
\STATE ~~Deduce $\bold{I}^*({\bf{S}}(T))$ for the state evolution $\bold{v}(T) = \bold{S}(T)$ according to \textbf{Algorithm $2$},\\
\STATE ~~Deduce the optimal summation of payoffs ${U^{*}}\left( {{\bf{S}}}(T) \right)$ according to~\eqref{eq:34},\\
\STATE ~~Record the optimized sets of decisions as ${\bf{L}}(T)|_{{\bf{S}}(T)} = \{ (\bold{S}(T),\bold{I}^*({\bf{S}}(T)))\}$.
\STATE \textbf{Repeat:}
\STATE ~$t=t-1$;
\STATE ~\textbf{for all} allowable battery state~${\bf{S}}(t)$
\STATE ~~\textbf{for all} allowable state evolution~${\bf{v}}(t)$\\
\STATE ~~~Deduce $\bold{I}^*(\bold{v}(t))$ for the $\bold{v}(t)$ according to \textbf{Algorithm $2$},\\
\STATE ~~~Find $U^{*}\left( {{\bf{S}}\left( t + 1 \right)} \right)$ with ${{\bf{S}}\left( {t + 1} \right)} = {{\bf{S}}\left( {t} \right)} - \bold{v}(t)$ in the last repetition.
\STATE ~~Deduce the optimal summation of payoffs $U^{*}\left( {{\bf{S}}\left( t \right)} \right)$ according to~\eqref{eq:33},\\
\STATE ~~Mark the optimal decision as $(\bold{v}^{*}(t),\bold{I}^*(\bold{v}^{*}(t))) = \arg U^{*}\left( {{\bf{S}}\left( t \right)} \right)$ according to~\eqref{eq:33},\\
\STATE ~~Record the optimized sets of decisions as ${\bf{L}}(t)|_{{\bf{S}}(t)} = \{(\bold{v}^{*}(t),\bold{I}^*(\bold{v}^{*}(t))),{\bf{L}}(t+1)\}$.
  \STATE \textbf{Until:} $t=1$.
\STATE \textbf{Return:} The optimal set of decisions ${\bf{L}}(1)|_{{\bf{S}}(1)}$.
\end{algorithmic}
\end{algorithm}

\section{Optimized Joint Power Splitting and Battery Operation Design with Causal Channel State Information}

In the previous section, the optimal joint PS and battery operation design is derived with the knowledge of non-causal CSI, which reveals a theoretical bound for the proposed harvest-use-store PS relaying strategy. Considering that the non-causal CSI is hard to obtain in practice, two optimized joint PS and battery operation designs are presented in this section when the CSI is known only causally. First, a Finite-state Markov Chain model approach is proposed, which extends the proposed layered optimization method by making use of the statistical properties of the CSI. Then, a greedy method approach is designed, which requires no information about the CSI during the subsequent time slots.
\vspace*{-0.5em}
\subsection{A Finite-state Markov Chain Model Approach}

In the proposed layered optimization method, both the information transfer design and the energy level variation design need to be optimized. Note that, when the CSI is known only causally, the information transfer design needs to be optimized at the beginning of each time slot, while the energy level variation design needs to be optimized before transmission because it is coupled over time. However, the optimized information transfer design is required in the embedded problem before transmission. To ensure this, the optimization of the joint PS and battery level variation design with causal CSI is carried out in two steps. Before transmission, similar to the non-causal case, both the embedded information transfer design and the overall energy level variation design are optimized with the modeled link gains. Then, at the beginning of each transmission, the information transfer design is updated and optimized with accurate link gains obtained from the causal CSI.

Before transmission, an $m$-state First-order Markov chain model~\cite{markov} is adopted to represent link gains using the knowledge of the distributions of channels, which has been proved to be an accurate model for slow fading channels in~\cite{markov}. As a result, the continuous variable representing the amplitude of each link gain is quantized by mapping it to $m$ non-overlapping intervals, each of which is called a state. Different methods to obtain the boundary values of these intervals have been concluded in~\cite{markov}, and the equal probable steady state method is adopted in this paper, where the steady state probabilities of the next states are equal. Denoting the sets of quantized link gains as ${{\bf{\zeta }}_{S,k}} = \left\{ {h_k^1,h_k^2, \cdots h_k^m} \right\}$ for $h_k(t), \forall t$ and ${{\bf{\zeta }}_{k,D}} = \left\{ {g_k^1,g_k^2, \cdots g_k^m} \right\}$ for $g_k(t), \forall t$, the link gains in Problem (P$1$) are replaced by
\begin{equation}\label{eq:f1}
\begin{array}{l}
| h_k(t) | = |\varsigma_{S,k}(t)|, \forall k, \forall t\\
| g_k(t) | = |\varsigma_{k,D}(t)|, \forall k, \forall t
\end{array}
\end{equation}
where $|\varsigma_{S,k}(t)| \in {{\bf{\zeta }}_{S,k}}$ and $|\varsigma_{k,D}(t)| \in {{\bf{\zeta }}_{k,D}}$. Further, to give a distinct demonstration of the finite-state Markov chain model approach, a block fading channel model is considered in this paper, i.e., values of link gains at different time slots are independent. Note that, the correlated channel models can also be analyzed by using the finite-state Markov chain models presented in~\cite{markov}. As a result, the transition probability of states between two adjacent time slots can be given by
\begin{equation}\label{eq:71add}
\begin{array}{l}
p(\varsigma_{S,k}(t)|\varsigma_{S,k}(t-1)) = p(\varsigma_{S,k}(t)), \forall k, \forall (t \ne 1) \\
p(\varsigma_{k,D}(t)|\varsigma_{k,D}(t-1)) = p(\varsigma_{k,D}(t)). \forall k, \forall (t \ne 1)
\end{array}
\end{equation}
According to the equal probable steady state method, the probability of each state at time slot $t$ is given by
\begin{equation}\label{eq:72add}
p(\varsigma_{S,k}(t)) = p(\varsigma_{k,D}(t)) = 1/m. \forall k, \forall t\\
\end{equation}

To optimize the energy level variation design~(i.e., $\bold{v}(t),\forall t$) before transmission, we follow a similar approach as the proposed layered optimization method in the previous section. First, the following components are defined.
\begin{enumerate}
\setcounter{enumi}{5}
\item channel state: The channel state is defined as the quantized link gains of all channels, which is described by ${\bf{H}}(t) = \{|\varsigma_{S,1}(t)|,\cdots,|\varsigma_{S,K}(t)|,|\varsigma_{1,D}(t)|,\cdots,|\varsigma_{K,D}(t)|\}$ at time slot $t$.
\end{enumerate}
With a channel state ${\bf{H}}(t)$, the payoff defined in~\eqref{eq:31} can be rewritten as $R'({\bold{S}}(t),{\bf{H}}(t),\bold{D}(t))$, which is obtained via replacing $|h_k(t)|,|g_k(t)|$ by $|\varsigma_{S,k}(t)|,|\varsigma_{k,D}(t)|$ in~\eqref{eq:31}. As a result, the expected summation of payoffs can be utilized to help the system make decisions.

\begin{enumerate}
\setcounter{enumi}{6}
\item expected summation of payoffs: Similar to the definition in~\eqref{eq:32}, the expected summation of payoffs is defined by
    \begin{equation}\label{eq:72}
  \begin{array}{l}
  {\tilde U}({\bf{S}}(t),{\bf{H}}(t),{{\bf{{D}}}}(t),t) = {R'}({\bf{S}}(t),{\bf{H}}(t),{{\bf{{D}}}}(t),t) \\
   + \sum\limits_{{\bf{H}}(t + 1)} \left\{p \left( {{\bf{H}}(t + 1)} \right)\tilde U\left({\bf{S}}(t + 1),{\bf{H}}(t + 1),\right.\right.\\
  \quad\quad \quad \quad \quad \left. \left. {\bf{D}}(t + 1),(t+1) \right) \right\} ,\; \forall (t \ne T)\\
  {\tilde U}\left({\bf{S}}\left(T\right),{\bf{H}}(T),{{\bf{{D}}}}\left(T\right),T\right) \\ = {R'}\left({\bf{S}}\left(T\right),{\bf{H}}(T),{\bf{{D}}}\left(T\right),T\right).
  \end{array}
  \end{equation}
\end{enumerate}
The goal is to maximize the expected summation of payoffs. First, the embedded optimization problem to derive $R'({\bold{S}}(t),{\bf{H}}(t),\bold{D}(t),t)$ can be solved by using \textbf{Algorithm $2$}, then the programming problem to derive the optimized energy level variation design can be solved by using backward induction. In particular, the Bellman equation is described by~\cite{dyn}
\begin{equation}\label{eq:73}
\begin{array}{l}
{{\tilde U}^*}\left( {{\bf{S}}\left( t \right)},{\bf{H}}(t) ,t\right) = \mathop {\max}\limits_{{\bf{v}}(t)} \{ R'({\bold{S}}(t),{\bf{H}}(t),\bold{D}(t),t) \\
+ \sum\limits_{{\bf{H}}(t + 1)} \frac{1}{{m^{2K}}} {\tilde U}^*\left( {{\bf{S}}\left( {t + 1} \right)}, \bold{H}(t+1),(t+1)\right)\},\;\forall (t \ne T)
\end{array}
\end{equation}
\begin{equation}\label{eq:74}
{{\tilde U}^*}\left( {{\bf{S}}\left( T \right)},{\bf{H}}(T),T \right) = {R'}({\bf{S}}(T),{\bf{H}}(T),{\bf{S}}(T),{\bold{I}}^*(T),T),
\end{equation}
where the $\frac{1}{{m^{2K}}}$ in~\eqref{eq:73} is derived from~\eqref{eq:72add}, ${R'}({\bf{S}}(T),{\bf{H}}(T),{\bf{S}}(T),{\bold{I}}^*(T),T)$ in~\eqref{eq:74} is interpreted as $\bold{D}(T) = ({\bf{v}}(T),{\bold{I}}^*(T)) = ({\bf{S}}(T),{\bold{I}}^*(T))$, which implies that the stored power at all relays should be used up at the last time slot. Then, the procedure to optimize the energy level variation design before transmission is concluded in \textbf{Algorithm $4$}.
\begin{algorithm}
\caption{Expected Summation of Payoffs Optimization}\label{alg:Framwork}
\begin{algorithmic}[1]
\STATE \textbf{Input:} The set of states ${{\bf{\zeta }}_{S,k}}$, ${{\bf{\zeta }}_{k,D}},\forall k$ according to the equal probable steady state method~\cite{markov}.
\STATE \textbf{Initialization:}
\STATE ~Set the time index as $t=T$,\\
\STATE ~\textbf{for all} possible channel state~${\bf{H}}(T)$\\
\STATE ~~Step $4-7$ as in \textbf{Algorithm $3$} using~\eqref{eq:74}.
\STATE \textbf{Repeat:}
\STATE ~$t=t-1$;
\STATE ~\textbf{for all} possible channel state~${\bf{H}}(t)$\\
\STATE ~~Step $11-16$ as in \textbf{Algorithm $3$} using~\eqref{eq:73}.
\STATE \textbf{Until:} $t=1$.
\STATE \textbf{Return:} The look up table $\{\bold{v}^*({\bf{S}}(1)),{\bf{H}}(1)),\cdots,\bold{v}^*({\bf{S}}(T)),{\bf{H}}(T))\}$.
\end{algorithmic}
\end{algorithm}

By using \textbf{Algorithm $4$}, a look up table is established before transmission, which records the optimized energy level variation designs with each possible channel state and battery state~(i.e., $(\bold{S}(t),\bold{H}(t))$). After that, the information transfer design is updated and optimized with the causal CSI. In particular, at each transmission, the system first searched for the optimized energy level variation design in the look up table, where $\bold{H}(t)$ is obtained by mapping the accurate link gains to the quantized link gains. Next, the information transfer design at this time slot is optimized with the battery state, the accurate link gains and the located energy level variation design. This optimization problem has the same problem formulation as the embedded optimization problem in~\eqref{eq:add33} except that the optimized result of the energy level variation design is obtained from the look up table, which can be handled by using \textbf{Algorithm $2$}.

In summary, the two-step procedure to optimize the joint PS and battery operation design with a finite-state Markov chain model approach is described in \textbf{Algorithm $5$}.

\begin{algorithm}
\caption{Joint Power Splitting and Battery Operation Optimizing with a Finite-state Markov Channel Model Approach} \label{alg:Framwork}
\begin{algorithmic}[1]
\STATE \textbf{Input:} Causal CSI~(i.e., $h_k(t),g_k(t), \forall k$ at time slot $t$), and the look up table through \textbf{algorithm $4$}.
\STATE \textbf{Initialization:}
\STATE ~~Set the time index as $t=1$, and the battery state as $\bold{S}(t)=\{0,\cdots,0\}$,
\STATE ~~Set the optimized energy level variation solution at the last time slot as~$\bold{v}^*(t-1) = \{0,\cdots,0\}$.
\STATE \textbf{Repeat:}
\STATE ~~Update the battery state $\bold{S}(t)$ according to~\eqref{eq:add3} with $\bold{v}^*(t-1)$,
\STATE ~~Evaluate the channel state $\bold{H}(t)$ by mapping $h_k(t),g_k(t),\forall k$ to ${\bf{\zeta }}_{S,k},{\bf{\zeta }}_{k,D},\forall k$,
\STATE ~~Search for the optimal energy level variation design $\bold{v}^*(t)$ from the look up table,
\STATE ~~Deduce the optimized information transfer design $I^*(t)$ using \textbf{Algorithm $2$} with $\bold{S}(t)$ and $\bold{v}^*(t)$,
\STATE ~~Update $t=t+1$.
\STATE \textbf{Until:} $t=T$.
\STATE \textbf{Return:} The optimized joint PS and battery operation design $(\bold{v}^*(t),I^*(t))$ at time slot $t$.
\end{algorithmic}
\end{algorithm}
\vspace*{-0.5em}
\subsection{A Greedy Method Approach}

To utilize the statistical properties of CSI, a look up table is established in the above finite-state Markov chain model approach, which requires high computation complexity when the number of channel
states is large. To lower the complexity, a greedy method that requires no information of the subsequent time slots is proposed. The key idea is to maximize the throughput at each transmission, and ignore the effects on the subsequent time slots. As a result, Problem (P$1$) is approached by time decoupled throughput optimization problems at each time slot, where the formulated optimization problem at time slot $t$ is given by
\begin{equation}\label{eq:54add}
\begin{array}{l}
\left( {P5} \right): \; \mathop {\max }\limits_{\bold{v}(t),\bold{I}(t) } \quad R(\bold{v}(t),\bold{I}(t))= \frac{1}{2} {\log \left( {1 + SNR\left( t \right)} \right)}\\
s.t. \quad 0 \le \lambda _{k,I}\left( t \right) \le 1 + \min \left\{ {0,\frac{{{v_k}\left( t \right)}}{{{\eta _1}{\eta _2}P{{\left| {{h_k}\left( t \right)} \right|}^2}}}} \right\},\;\forall k\\
{v_k}\left( t \right) = {n(t)\frac{B_{max}}{L}, n(t) \in \left\{ - L,\cdots,0,\cdots,L \right\}},\;\forall k\\
{B_k}\left( t \right)  - B_{max} \le {v_k}\left( t \right) \le {B_k}\left( t \right),\;\forall k\\
{B_k}\left( t \right) = B_k\left( t-1 \right) - v_k^*(t-1), \;\forall k
\end{array}
\end{equation}
where $v_k^*(t-1)$ is the optimized energy level variation design at the last time slot, and the fourth constraint implies that the energy level variation design results in evolution of battery states, which is derived from~\eqref{eq:add3}. Since the goal is to maximize $R(\bold{v}(t),\bold{I}(t))$, while the effects of design on the subsequent time slots are not evaluated, the following insight can be given to optimize the energy level variation design.

\begin{theorem}
If all batteries are empty initially, the optimized joint PS and battery operation design using the greedy method is to use up the harvested energy at each transmission.
\end{theorem}
\begin{proof}
First, similar to the formulation of~\eqref{eq:add33}, since $R(\bold{v}(t),\bold{I}(t))$ is monotonically increasing in terms of $SNR(t)$, Problem (P$5$) is equivalent to the SNR maximization problem~\cite{dis}.

The theorem is proved step by step. At the first time slot, all batteries are empty~(i.e., $B_{k}(1)=0,\forall k$). Since the power provided by battery discharging cannot exceed the power stored in the battery~(i.e., $b_{k,F}(t) \le B_{k}(t),\forall k,\forall t$), the power discharged at time slot $1$ is
\begin{equation}\label{eq:f12}
b_{k,F}(1) = 0.\forall k
\end{equation}
To maximize the optimization objective $SNR(1)$, the PS ratio for battery charging should satisfy
\begin{equation}\label{eq:f13}
\lambda_{k,B}(1) = 0,\forall k
\end{equation}
because otherwise $SNR(1)$ can be further improved by reassigned the PS ratios that: $\lambda'_{k,B}(1) = 0$, $\lambda_{k,F}(1)$ remains the same, $\lambda'_{k,I}(1) = \lambda_{k,I}(1) + \lambda_{k,B}(1)$, which can be seen according to~\eqref{eq:2add},~\eqref{eq:6},~\eqref{eq:7}, and~\eqref{eq:9}. Thus, no signal power is spit for battery charging and the harvested energy is used up at this transmission. As a result, each battery stays empty at the next time slot.
\begin{equation}\label{eq:f14}
{B_k}(2) = {B_k}\left( 1 \right) + \eta_1 \eta_2 {\lambda_{k,B}}\left( 1 \right){\left| {{y_k}\left( 1 \right)} \right|^2} - {b_{k,F}}\left( 1 \right)=0, \forall k
\end{equation}
which is derived from~\eqref{eq:8}.

Assuming that the harvested energy is used up at time slot $(t-1)$ and each battery stays empty at time slot $t$~(i.e., ${B_k}(t) = 0. \forall k$), then similar to the above analysis, it can be derived that
\begin{equation}\label{eq:f15}
\begin{array}{l}
b_{k,F}(t) = 0,\forall k\\
\lambda_{k,B}(t) = 0.\forall k
\end{array}
\end{equation}
In conclusion, the optimized joint PS and battery operation design using the greedy method is to use up the harvested energy at each time slot.
\end{proof}

Based on Theorem $1$ and the definition of energy level variation in~\eqref{eq:12}, it can be derived that
\begin{equation}\label{eq:f16}
{v_k}^*\left( t \right) = 0.\;\forall k, \forall t
\end{equation}
Substituting~\eqref{eq:f16} in Problem (P$5$), the optimization problem can be rewritten as
\begin{equation}\label{eq:54add}
\begin{array}{l}
\left( {P6} \right): \; \mathop {\max }\limits_{\bold{I}(t) } \quad \frac{1}{2} {\log \left( {1 + SNR_1\left( t \right)} \right)}\\
s.t. \quad 0 \le \lambda _{k,I}\left( t \right) \le 1,\;\forall k
\end{array}
\end{equation}
where $SNR_1\left( t \right)$ is obtained via replacing ${v_k}\left( t \right)$ by $0$ in~\eqref{eq:9}. As a result, Problem (P$6$) has the same problem formulation as the embedded optimization problem in~\eqref{eq:add33} and can be solved by using \textbf{Algorithm $2$}. In summary, the procedure to optimize the joint PS and battery operation design with a greedy method is described in \textbf{Algorithm $6$}, which solves Problem (P$6$) at different time slots step by step.

\begin{algorithm}
\caption{Joint PS and Battery Operation Optimizing with a Greedy Method} \label{alg:Framwork}
\begin{algorithmic}[1]
\STATE \textbf{Input:} Causal CSI~(i.e., $h_k(t),g_k(t), \forall k$ at time slot $t$).
\STATE \textbf{Initialization:}
\STATE ~~Step $3,4$ as in \textbf{Algorithm $5$}.
\STATE \textbf{Repeat:}
\STATE ~~Set the optimized energy level variation design as $\bold{v}^*\left( t \right) = \{0,\cdots,0\}$,
\STATE ~~Step $6,9,10$ as in \textbf{Algorithm $5$} to solve Problem (P$6$).
\STATE \textbf{Until:} $t=T$.
\STATE \textbf{Return:} The optimized joint PS and battery operation design $(\bold{v}^*(t),I^*(t))$ at time slot $t$.
\end{algorithmic}
\end{algorithm}

Note that, Theorem $1$ indicates that the battery is not utilized in the proposed greedy method, such that a lower bound on performance of the proposed harvest-use-store PS relaying strategy is provided. Moreover, the solution in Theorem $1$ for the wireless-powered cooperative relaying is different from the solution for the conventional cooperative relaying~\cite{dis}, where not all relays should use up their available power. In particular, a tradeoff between signal amplification and noise amplification needs to be reached in the amplify-and-forward relaying, because the transmit power not only amplifies the useful signal but also the noise. This tradeoff is realized by controlling the transmit power at the relays in conventional cooperative relaying, thus only the relays with good channel conditions should use up the transmit power. However, this tradeoff is reached by optimizing the PS ratios in the wireless-powered case, because the harvested energy obtained through PS is the only power source for relays. As a result, the reason to store a part of the harvested energy is to improve the performance at subsequent time slots and reach a better overall throughput.
\vspace*{-0.5em}
\subsection{Computational Complexity}

The computational complexity of the proposed two online algorithms and exhaustive searching are compared in this subsection. Note that, all the three algorithms utilize \textbf{Algorithm $2$} to solve the ``embedded problem", thus, their computational complexity can be compared by evaluating the times \textbf{Algorithm $2$} is used. Assuming that the computational complexity of \textbf{Algorithm $2$} is $C_2$, the total number of time slots is $T$, the number of channel states is $n_c$, the number of battery states is $N$. According to the constraints $C2$ and $C3$ in~\eqref{eq:15}, the number of battery operations is also $N$. As a result, the computational complexities of the three algorithms can be seen in Table I. It can be seen that the complexity of exhaustive searching is of order $O((N^2)^T)$, which is because the evaluation of battery operations through $T$ time slots are coupled over time. The complexity of algorithm $5$ is of order $O(N^2)$ because dynamic programming adopted in this algorithm is a time decoupled procedure. Besides, $n_c$ is due to channel state information prediction through Markov chain model, which makes dynamic programming feasible. Moreover, the complexity of \textbf{Algorithm $6$} is of order $O((N^2)^0)$, which is because the optimized battery operation in greedy algorithm is set to be $0$ for each relay at each time slot according to Theorem $1$. In conclusion, \textbf{Algorithm $5$} enjoys a lower computational complexity than exhaustive searching, and \textbf{Algorithm $6$} is of lower computational complexity than \textbf{Algorithm $5$}.

\begin{table}
\caption{Computational Complexity Comparison} \center
\begin{tabular}{c c c c}
\hline
 & Computational Complexity\\
\hline
Exhaustive searching & $O(C_2({N}^2)^T)$\\
Algorithm 5 & $O(n_cC_2T(N^2))$\\
Algorithm 6 & $O(C_2T)$\\\hline
\end{tabular}
\end{table}

\section{Numerical Results}

In this section, the proposed harvest-use-store PS relaying strategy is evaluated with numerical results. The two-hop channels are modeled by ${{\vert\rho_{Sk}\vert}^2}d_{Sk}^{-2}K$ and ${{\vert\rho_{kD}\vert}^2}d_{kD}^{-2}$ respectively, where $k$ is the index for relays. To simplify simulations, the distance between the source node and each relay is normalized as $d_{Sk} = 1$m, while the distance between each relay and the destination node is $d_{kD} = 5$m, a similar simulation scenario can be seen in~[9]. Besides, ${\vert\rho_i\vert}$ denotes the short-term channel fading, and is assumed to be Rayleigh distributed. ${{\vert\rho_i\vert}^2}$ follows the exponential distribution with unit mean. The noise powers at each relay and the destination are set as $\sigma_b^2 = \sigma_D^2 = \sigma = 1$. The average signal-to-noise ratio~(SNR) is defined as $\frac {P}{\sigma}$, where $P$ is the maximum transmit power at the source node. Besides, the energy conversion efficiency is set as $\eta_1 = 0.4$, the storage efficiency is $\eta_2 = 0.8$, and the coefficient for the battery size is $\alpha = 1$. For the purpose of presentation, the proposed three joint PS and battery operation designs are termed as the optimal design, the general design, and the greedy design, respectively.
\vspace*{-0.5em}
\subsection{Convergence of the Proposed Designs}

Convergence behavior of the adopted iterative algorithms is depicted in Fig. \ref{conv_1}, where the ``$a-b$ iterations" in this figure denotes $a$ iterations in the alternating optimization and $b$ iterations in the Dinkelbach algorithm. Note that, these two iterations in \textbf{Algorithm $1$} and \textbf{Algorithm $2$} are adopted in all three proposed designs, thus only the convergence behavior of the optimal design is presented due to limited space. We set the number of relays as $K = 2$, the number of energy levels as $L + 1 = 10$, the total number of time slots as $T = 5$, and the number of Markov states as $m = 3$. It's revealed that the alternating optimization algorithm converges within $5$ iterations, and so is the Dinkelbach algorithm.
\begin{figure}[!htb]
\centering\vspace*{-1em}
\includegraphics[width = 3.4 in]{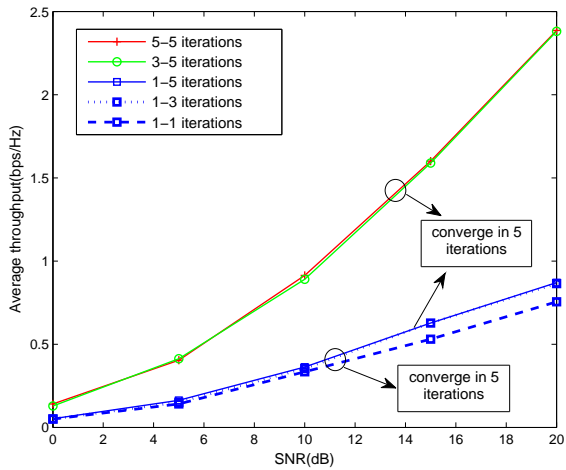}
\caption{Convergence of the proposed optimal design (\emph{K} = $2$, \emph{T} = $5$, and $L = 9$).}\label{conv_1}
\end{figure}
\vspace*{-0.5em}
\subsection{Performance Comparison of the Proposed Designs}

To demonstrate the performance comparison of the three proposed designs, Fig. \ref{comp_pro} is provided as below. This figure demonstrates the average throughput performance of different joint PS and battery operation designs versus the SNR, where the number of relays is $K = 2$, the number of energy levels is $L + 1 = 5$, the total number of time slots is $T = 10$, and the number of Markov states is $m = 3$. The exhaustive searching will determine the optimal joint PS and battery operation design with non-causal CSI numerically. First, the performance of the proposed optimal design is indistinguishable from the optimal one derived from exhaustive searching, which provides an upper bound on performance for the proposed strategy. Next, the performance of the proposed greedy design provides a lower bound on performance for the proposed strategy which does not need a battery. Moreover, the results reveal that by utilizing the non-causal CSI, or to a lesser extent the statistical information of CSI, the throughput performance of the proposed strategy can be improved.
\begin{figure}[!htb]
\centering\vspace*{-1em}
\includegraphics[width = 3.4 in]{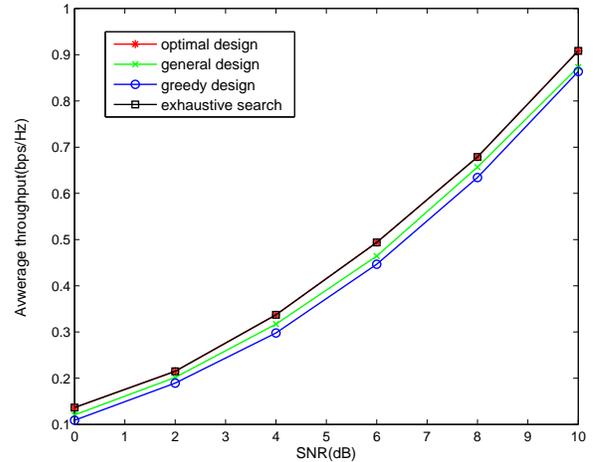}
\caption{Average throughput performance with different joint PS and battery operation designs (\emph{K} = $2$, \emph{T} = $10$, $m = 3$ and $L = 4$).}\label{comp_pro}
\end{figure}
\vspace*{-0.5em}
\subsection{Performance Comparison with Conventional Strategies}

To reveal the advantages of using distributed beamforming in the proposed strategy, performance comparison using different precoding techniques are depicted in Fig. \ref{pre} in terms of system throughput. We set the number of relays to be $K = 3$, the number of energy levels to be $L + 1 = 5$, and the total number of time slots to be $T = 5$. It can be seen that the proposed designs with distributed beamforming outperform the optimized joint PS and battery operation designs with both the best relay selection and random relay selection.
\begin{figure}[!htb]
\centering\vspace*{-1em}
\includegraphics[width = 3.4 in]{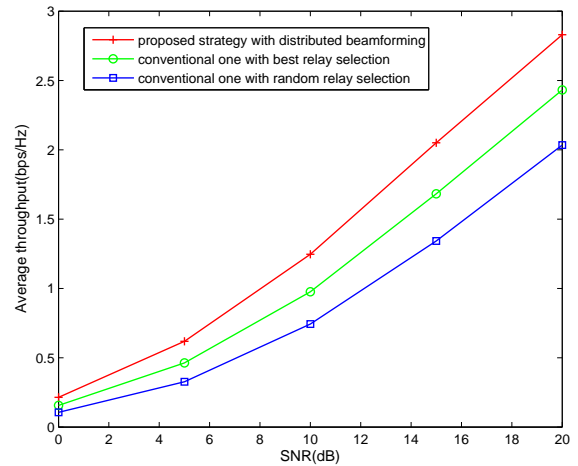}
\caption{Average throughput performance with different precoding techniques, with \emph{K} = $3$, \emph{T} = $5$ and $L = 4$.}\label{pre}
\end{figure}

To reveal the advantages of using PS receiving architecture in the proposed strategy, performance comparisons using different energy harvesting receiving architectures are illustrated in Fig. \ref{comp_eh} and Fig. \ref{delay}. Specifically, the average throughput performance versus SNR has been depicted in Fig. \ref{comp_eh}, where the number of relays is $K = 2$, the number of energy levels is $L + 1 = 5$, and the number of time slots is $T = 10$. It is revealed that our proposed strategy outperforms time switching-based relaying in terms of system throughput.

\begin{figure}[!htb]
\centering\vspace*{-1em}
\includegraphics[width = 3.4 in]{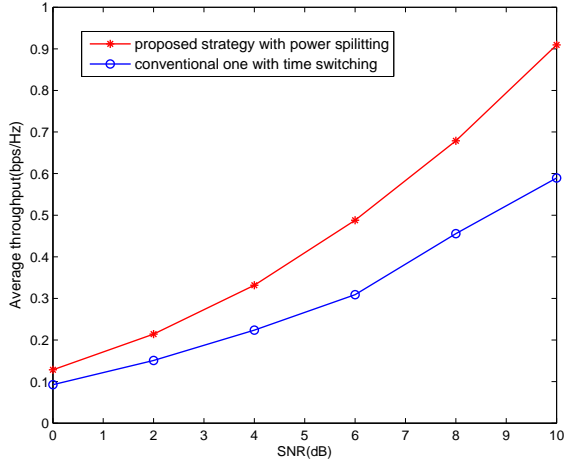}
\caption{Average throughput performance with different energy harvesting receiving architectures (\emph{K} = $2$, \emph{T} = $10$, and $L = 4$).}\label{comp_eh}
\end{figure}

Moreover, performance comparison in terms of delay constraint services is demonstrated in Fig. 6, where the number of relays is $K = 2$, the number of energy levels is $L + 1 = 5$, the bandwidth is $1$MHz, and the average SNR is $10dB$. In this figure, the minimum required slots to finish data transmission of a delay constraint service is depicted. It's revealed that, the proposed strategy with PS consumes fewer time slots compared with the time switching-based relaying, which indicates that the proposed strategy is more suitable for applications with critical delay constraints.

\begin{figure}[!htb]
\centering\vspace*{-1em}
\includegraphics[width = 3.4 in]{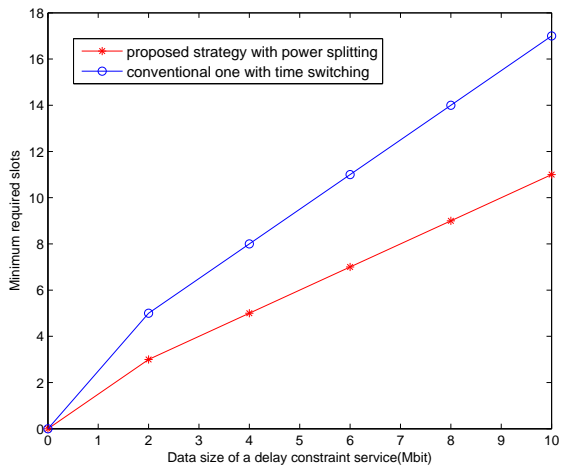}
\caption{Minimum required time slots versus data size of a delay constraint services (\emph{K} = $2$, $L = 4$).}\label{delay}
\end{figure}

To show the advantages of using harvest-use-store model in the proposed strategy, performance comparison using different power management models are illustrated in Fig. \ref{comp_pm} in terms of system throughput, where the number of relays is $K = 2$, the number of energy levels is $L + 1 = 10$, and the number of time slots is $T = 5$. Note that, the proposed strategy outperforms the one with harvest-store-use model because the proposed strategy avoids unnecessary storage loss at relay batteries~\cite{sto}. Moreover, the proposed strategy outperforms the one with harvest-use model because the harvested energy can be accumulated for future usage, which realizes a more efficient utilization of harvested energy~\cite{swipt5}.

\begin{figure}[!htb]
\centering\vspace*{-1em}
\includegraphics[width = 3.4 in]{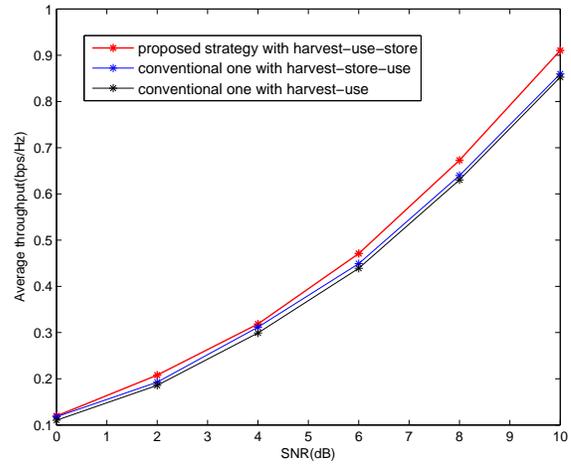}
\caption{Average throughput performance with different power management models (\emph{K} = $2$, \emph{T} = $5$, and $L = 9$).}\label{comp_pm}
\end{figure}

\section{Conclusion}

In this paper, to support an efficient utilization of harvested energy to improve throughput for wireless-powered multi-relay cooperative networks, a harvest-use-store power splitting (PS) relaying strategy with distributed beamforming has been researched. Since the formulated throughput maximization problem is intractable, a layered optimization method has been proposed to decompose the joint PS and battery operation design in two layers, which transforms the intractable optimization problem into a dynamic programming problem with a subproblem requiring optimization embedded in it. The layered optimization method has been implemented in the non-causal channel state information~(CSI) case, which leads to a theoretical bound of the proposed strategy, and extended to the general causal CSI case. To achieve a tradeoff between performance and complexity, a greedy method has been proposed to optimize the joint PS and battery operation design with causal CSI. Simulation results have shown that the proposed harvest-use-store PS-based relaying strategy outperforms time switching-based relaying strategy and conventional PS-based relaying strategy without energy accumulation. {\color{red}This work will be extended to full-duplex relay mode in our future works.}

\begin{appendices}
\vspace*{-0.5em}
\section{PROOF OF LEMMA 2}
\vspace*{-0.25em}
From~\eqref{eq:45}, the second-order derivative of the objective function in Problem (P$4$) is derived as
\begin{equation}\label{eq:a1}
\begin{array}{*{20}{l}}
\begin{array}{l}
\frac{{{\partial ^2}{J_2}}}{{\partial {x_j}^2}} =  - 2\sqrt {{\eta _1}{{\left| {{g_j}} \right|}^2}} \frac{{{a_j}\sigma _b^2}}{{{x_j}^3}}\left[ {\sqrt {{\eta _1}{{\left| {{g_j}} \right|}^2}} \left( {q + 1} \right)} \right.\\
\quad \quad \quad \quad \quad \quad \quad \quad \left. { + \frac{{\sum\limits_{k \ne j} {\sqrt {{\eta _1}{{\left| {{g_k}} \right|}^2}\left( {{a_k} - {x_k}} \right)\left( {1 - \frac{{\sigma _b^2}}{{{x_k}}}} \right)} } }}{{\sqrt {\left( {{a_j} - {x_j}} \right)\left( {1 - \frac{{\sigma _b^2}}{{{x_j}}}} \right)} }}} \right]
\end{array}\\
{ - \frac{1}{2}\sqrt {{\eta _1}{{\left| {{g_j}} \right|}^2}} {{\left( {\frac{{{a_j}\sigma _b^2}}{{{x_j}^2}} - 1} \right)}^2}\frac{{\sum\limits_{k \ne j} {\sqrt {{\eta _1}{{\left| {{g_k}} \right|}^2}\left( {{a_k} - {x_k}} \right)\left( {1 - \frac{{\sigma _b^2}}{{{x_k}}}} \right)} } }}{{{{\left( {{a_j} - {x_j}} \right)}^{\frac{3}{2}}}{{\left( {1 - \frac{{\sigma _b^2}}{{{x_j}}}} \right)}^{\frac{3}{2}}}}} \le 0.}
\end{array}
\end{equation}
Since $\frac{{{\partial ^2}{(F_1(x_j)-qF_2(x_j))}}}{{\partial {x_j}^2}}$ is non-positive, the optimization objective in Problem (P$4$) is a concave function in terms of $x_j$. This completes the proof of Lemma~$2$.

\vspace*{-0.5em}
\section{PROOF OF LEMMA 4}
\vspace*{-0.25em}
At each iteration of \textbf{Algorithm $2$}, one of the $x_j,\forall j$ is updated, whose value is derived by using \textbf{Algorithm $1$}. Note that, the optimal solution to Problem (P$3$) is derived from \textbf{Algorithm $1$}, thus when updating $x_j$, the following relationship is derived that
\begin{equation}\label{eq:33add}
\begin{array}{l}
\mathop {\max }\limits_{ {x_j} } \frac{F_1(x_j)}{F_2(x_j)} = \frac{F_1(x'_j)}{F_2(x'_j)}=J(x_1,\cdots,x_j',\cdots,x_K) \\
> J(x_1,\cdots,x_j,\cdots,x_K),\; \forall j
\end{array}
\end{equation}
which implies that the optimization objective $J$ in Problem (P$2$) increases after each iteration. Moreover, there exist an upper bound for $J$ that
\begin{equation}\label{eq:a3}
J_{upper} = \frac{{{{\left( {\sum\limits_{k = 1}^K {\sqrt {{\eta _1}{{\left| {{g_k}} \right|}^2}\left( {{a_k} - \sigma _b^2} \right)\left( {1 - \frac{{\sigma _b^2}}{{{b_k}}}} \right)} } } \right)}^2}}}{{\sum\limits_{k = 1}^K {{\eta _1}{{\left| {{g_k}} \right|}^2}\left( {{a_k} - \sigma _b^2} \right)\frac{{\sigma _b^2}}{{{b_k}}}}  + \sigma _D^2}},
\end{equation}
where ${b_k} = P{\left| {{h_k}\left( t \right)} \right|^2} + \sigma _b^2$. This upper bound is formulated by substituting $\lambda_{k,I} = 1$ and $\lambda_{k,F} + \lambda_{k,B}= 1$ into~\eqref{eq:9} and simplifying~\eqref{eq:9} in a manner similar to the formulation of~\eqref{eq:43}, which is the theoretical ideal solution but may not be practical as explained in~\cite{prac}. The practical restriction $\lambda_{k,I} + \lambda_{k,F} + \lambda_{k,B}= 1$ is adopted in our paper, and thus the optimized performance cannot outperform the upper bound. In conclusion, after each iteration in\textbf{ Algorithm $2$}, the optimization objective in Problem (P$2$) is increased and finally approaches the upper bound. Thus, \textbf{Algorithm $2$} converges.

\end{appendices}

\begin{IEEEbiography}[{\includegraphics[width=1in]{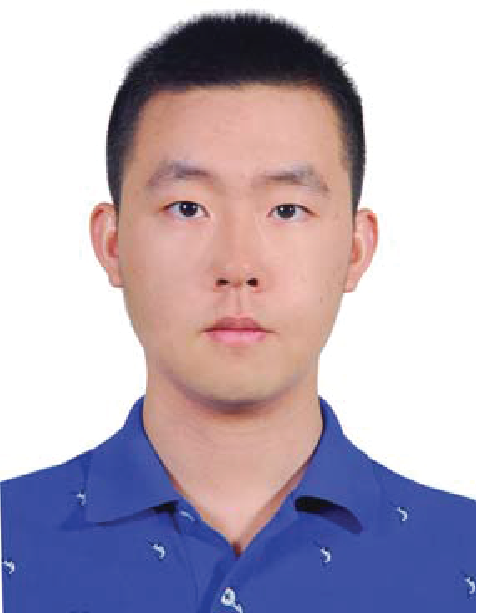}}]{Zheng Zhou}
received the B.S. degree in Information Engineering from Beijing University of
Posts and Telecommunications (BUPT), Beijing,
China, in 2012. He is currently working toward the
Ph.D. degree at BUPT. His research interests include
simultaneous information and power transfer and cloud radio access
networks.
\end{IEEEbiography}

\begin{IEEEbiography}[{\includegraphics[height=1.25in]{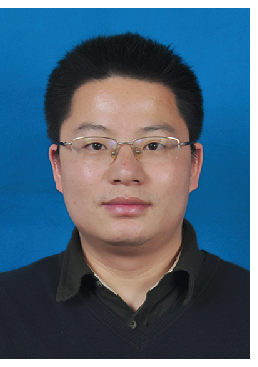}}]{Mugen Peng}
(M'05--SM'11) received the B.E. degree in Electronics Engineering
from Nanjing University of Posts \& Telecommunications, China in
2000 and a PhD degree in Communication and Information System from
the Beijing University of Posts \& Telecommunications (BUPT), China
in 2005. After the PhD graduation, he joined in BUPT, and has become
a full professor with the school of information and communication
engineering in BUPT since Oct. 2012. During 2014, he is also an
academic visiting fellow in Princeton University, USA. He is leading
a research group focusing on wireless transmission and networking
technologies in the Key Laboratory of Universal Wireless
Communications (Ministry of Education) at BUPT, China. His main
research areas include wireless communication theory, radio signal
processing and convex optimizations, with particular interests in
cooperative communication, radio network coding, self-organization
networking, heterogeneous networking, and cloud communication. He
has authored/coauthored over 40 refereed IEEE journal papers and
over 200 conference proceeding papers.

Dr. Peng is currently on the Editorial/Associate Editorial Board of
\emph{IEEE Communications Magazine}, \emph{IEEE Access},
\emph{International Journal of Antennas and Propagation (IJAP)},
and \emph{China Communications}. He has been the guest leading editor
for the special issues in \emph{IEEE Wireless Communications},
\emph{IJAP} and \emph{International Journal of Distributed Sensor
Net- works (IJDSN)}. He received the 2014 IEEE ComSoc AP Outstanding
Young Researcher Award, and the Best Paper Award in IEEE WCNC 2015, WASA 2015, GameNets 2014,
IEEE CIT 2014, ICCTA 2011, IC-BNMT 2010, and IET CCWMC 2009. He was
awarded the First Grade Award of Technological Invention Award in
Ministry of Education of China, and the Second Grade Award of
Scientific \& Technical Progress from China Institute of
Communications.
\end{IEEEbiography}

\begin{IEEEbiography}[{\includegraphics[width=1in]{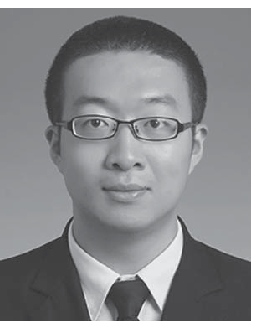}}]{Zhongyuan Zhao}
received the B.S. degree in applied
mathematics and the Ph.D. degree in communication
and information systems from Beijing University
of Posts and Telecommunications (BUPT), Beijing,
China, in 2009 and 2014, respectively. He is
currently a Lecturer with the Key Laboratory of Universal
Wireless Communication (Ministry of Education)
at BUPT. His research interests include network
coding, MIMO, relay transmissions, and large-scale
cooperation in future communication networks.
\end{IEEEbiography}

\begin{IEEEbiography}[{\includegraphics[width=1in]{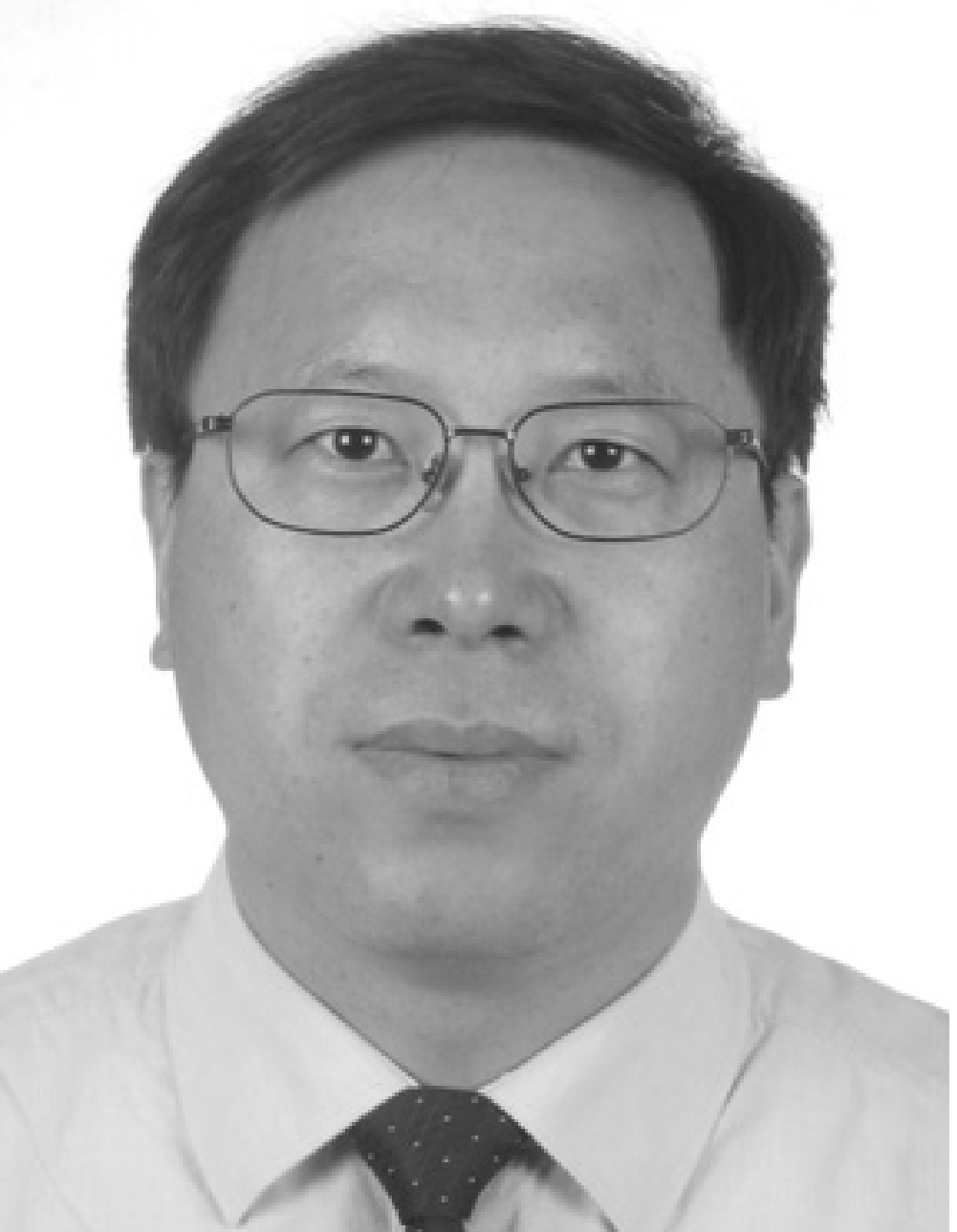}}]{Wenbo Wang}
is currently the dean of Telecommunication Engineering at Beijing
University of Posts and Telecommunications (BUPT). He received the
BS degree, the MS and Ph.D. Degrees from BUPT in 1986, 1989 and 1992
respectively. Now he is the Assistant Director of academic committee
of Key Laboratory of Universal Wireless Communication (Ministry of
Education) in BUPT. His research interests include radio
transmission technology, Wireless network theory, Broadband wireless
access and Software radio technology. Prof. Wenbo Wang has published
more than 200 journal and international conference papers and holds
12 patents and has published six books.
\end{IEEEbiography}

\begin{IEEEbiography}[{\includegraphics[width=1in]{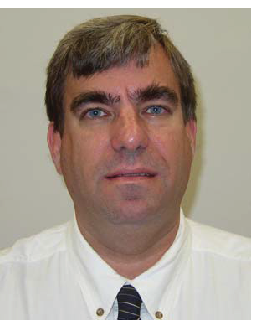}}]{Rick S. Blum}
(S'83--M'84--SM'94--F'05) received
the B.S. degree in electrical engineering from the
Pennsylvania State University in 1984 and the M.S.
and Ph.D. degrees in electrical engineering from
the University of Pennsylvania in 1987 and 1991,
respectively.

From 1984 to 1991, he was a member of technical
staff at General Electric Aerospace, ValleyForge, PA,
USA, and he graduated from GE¡¯s Advanced Course
in Engineering. Since 1991, he has been with the
Electrical and Computer Engineering Department at
Lehigh University, Bethlehem, PA, USA, where he is currently a Professor and
holds the Robert W. Wieseman Chaired Research Professorship in Electrical
Engineering. His research interests include signal processing for smart grid,
communications, sensor networking, radar and sensor processing. He is on
the editorial board for the \emph{Journal of Advances in Information Fusion} of the
International Society of Information Fusion. He was an Associate Editor for
\emph{IEEE TRANSACTIONS ON SIGNAL PROCESSING} and for \emph{IEEE COMMUNICATIONS
LETTERS}. He has edited special issues for \emph{IEEE TRANSACTIONS
ON SIGNAL PROCESSING}, \emph{IEEE JOURNAL OF SELECTED TOPICS IN SIGNAL
PROCESSING}, and \emph{IEEE JOURNAL ON SELECTED AREAS IN COMMUNICATIONS}.
He is a member of the SAM Technical Committee (TC) of the IEEE
Signal Processing Society. He was a member of the Signal Processing for
Communications TC of the IEEE Signal Processing Society and is a member of
the Communications Theory TC of the IEEE Communication Society. He was
on the awards committee of the IEEE Communication Society.

Dr. Blum is a former IEEE Signal Processing Society Distinguished Lecturer,
an IEEE Third Millennium Medal winner, a member of Eta Kappa
Nu and Sigma Xi, and holds several patents. He was awarded an ONR
Young Investigator Award and an NSF Research Initiation Award. His IEEE
Fellow Citation ``for scientific contributions to detection, data fusion and signal
processing with multiple sensors" acknowledges contributions to the field of
sensor networking.
\end{IEEEbiography}


\begin{thebibliography}{99}
\bibitem{2add2}
M. Peng, C. Wang, V. Lau, H. V.Poor, \textquotedblleft Fronthaul-constrained cloud radio access networks: insights and challenges,\textquotedblright \textit{IEEE Wireless Commun.}, vol.22, no.2, pp.152--160, April 2015.
{\color{red}
\bibitem{2add1}
B. Bangerter, S. Talwar, R. Arefi, and K. Stewart, \textquotedblleft Networks and Devices for the 5G Era,\textquotedblright \textit{IEEE Commun. Mag.}, vol. 52, no. 2, pp. 90--96, Feb. 2014.

\bibitem{2add4}
E. Hossain, M. Rasti, H. Tabassum, A. Abdelnasser, \textquotedblleft Evolution toward 5G multi-tier cellular wireless networks: An interference management perspective,\textquotedblright \textit{IEEE Wireless Commun.}, vol. 21, no. 3, pp. 118--127, June 2014.
}
\bibitem{natu1}
C. Huang, J. Zhang, P. Zhang, and S. Cui, \textquotedblleft Threshold-based transmissions for large relay networks powered by renewable energy,\textquotedblright ~in\textit{ Proc. IEEE Global Commun. Conf. (Globecom)}, Atlanta, GA, USA, Dec. 2013, pp. 1921--1926.

\bibitem{2add5}
Z. Ding, C. Zhong, D. Wing, M. Peng, H. A. Suraweera, R. Schober, and H. V. Poor, \textquotedblleft Application of smart antenna technologies in simultaneous wireless information and power transfer,\textquotedblright \textit{ IEEE Commun. Mag.}, vol. 53, no. 4, pp. 86--93, Apr. 2015.


\bibitem{2-1}
H. Tabassum, E. Hossain, A. Ogundipe, and D. I. Kim, \textquotedblleft Wireless-powered cellular networks: Key challenges and solution techniques,\textquotedblright \textit{
IEEE Commun. Mag.}, 2015 (to appear). [Online]. Available: http:
//wireless.skku.edu/english/UserFiles/File/final


\bibitem{htadd2}
X. Lu, P. Wang, D. Niyato, D. I. Kim, and Z. Han, \textquotedblleft Wireless networks with RF energy harvesting: A contemporary survey,\textquotedblright \textit{ IEEE Commun. Surveys \& Tutorials}, 2014, DOI: 10.1109/COMST.2014.2368999.

\bibitem{swipt1}
R. Zhang and C. K. Ho, \textquotedblleft MIMO broadcasting for simultaneous wireless information and power transfer,\textquotedblright \textit{ IEEE Trans. Wireless Commun.}, vol. 12, no. 5, pp. 1989--2001, May 2013.

\bibitem{eff}
K. Huang, E. Larsson, \textquotedblleft Simultaneous Information and Power Transfer for Broadband Wireless Systems,\textquotedblright \textit{ IEEE Trans. Sig. Proc.}, vol.61, no.23, pp.5972--5986, Dec. 2013.

\bibitem{delay}
I. Krikidis, S. Timotheou, S. Nikolaou, G. Zheng, D.W.K. Ng, R. Schober, \textquotedblleft Simultaneous wireless information and power transfer in modern communication systems,\textquotedblright \textit{ IEEE Commun. Mag.}, vol.52, no.11, pp.104--110, Nov. 2014.

\bibitem{swipt3}
A. A. Nasir, X. Zhou, S. Durrani, and R. A. Kennedy, \textquotedblleft Relaying protocols for wireless energy harvesting and information processing,\textquotedblright \textit{ IEEE Trans. Wireless Commun.}, vol. 12, no. 7, pp. 3622--3636, July 2013.

\bibitem{swipt4}
Z. Zhou, M. Peng, Z. Zhao and Y. Li, \textquotedblleft Joint power splitting and antenna selection in energy harvesting relay channels,\textquotedblright \textit{ IEEE Sig. Proc. Lett.}, vol. 22, no. 7, pp.823--827, July 2015.

\bibitem{swipt5}
Z. Ding, I. Krikidis, B. Sharif, and H. V. Poor, \textquotedblleft Wireless information and power transfer in cooperative networks with spatially random relays,\textquotedblright \textit{ IEEE Trans. Wireless Commun.}, vol. 13, no. 8, pp. 4440--4453, Aug. 2014.

\bibitem{htadd}
A. H. Sakr and E. Hossain, \textquotedblleft Analysis of multi-tier uplink cellular networks with energy harvesting and flexible cell association,\textquotedblright ~in\textit{ Proc. IEEE Global Commun. Conf. (Globecom)}, Austin, TX, USA, Dec. 2014, pp. 4525--4530.

\bibitem{ht3}
I. Krikidis, S. Timotheou, and S. Sasaki, \textquotedblleft RF energy transfer for cooperative networks: Data relaying or energy harvesting?,\textquotedblright \textit{ IEEE Commun. Lett.}, vol. 16, no. 11, pp.1772--1775, Nov. 2012.

\bibitem{ht4}
H. Yang, J. Lee, and T. Q. S. Quek, \textquotedblleft Heterogeneous cellular network with energy harvesting based D2D communication,\textquotedblright ~submitted to \textit{ IEEE Trans. Wireless Commun}. Available: [Online] http://www.researchgate.net/publication/263848668.

\bibitem{sto}
S. Sudevalayam and P. Kulkarni, \textquotedblleft Energy harvesting sensor nodes: Survey and implications,\textquotedblright \textit{ IEEE Commun. Surveys Tuts.}, vol. 13, no. 3, pp. 443--461, 2011.

\bibitem{ht5}
F. Yuan, Q. T. Zhang, S. Jin, and H. Zhu, \textquotedblleft Optimal harvest-use-store strategy for energy harvesting wireless systems,\textquotedblright \textit{ IEEE Trans. Wireless Commun.}, vol. 14, no. 2, pp. 698--710, Feb. 2015.

\bibitem{dis}
Y. Jing and H. Jafarkhani, \textquotedblleft Network beamforming using relays with perfect channel information,\textquotedblright \textit{ IEEE Trans. Inf. Theory}, vol. 55, no. 6, pp. 2499--2517, June 2009.

\bibitem{alt}
R. Wang and M. Tao, \textquotedblleft Joint Source and Relay Precoding Designs for MIMO Two-Way Relaying Based on MSE Criterion,\textquotedblright \textit{ IEEE Trans. Sig. Proc.}, vol.60, no.3, pp.1352--1365, March 2012.

\bibitem{noi}
L. Liu, R. Zhang, and K.-C. Chua, \textquotedblleft Wireless information and power transfer: A dynamic power splitting approach,\textquotedblright \textit{ IEEE Trans. Commun.}, vol. 61, no. 9, pp. 3990--4001, Sept. 2013.

\bibitem{batt3}
W. J. Huang, Y. W. P. Hong, and C. C. J. Kuo, \textquotedblleft Lifetime maximization for amplify-and-forward cooperative networks,\textquotedblright \textit{ IEEE Trans. Wireless Commun.}, vol. 7, no. 5, pp. 1800--1805, May 2008.

\bibitem{fractional}
W. Dinkelbach, \textquotedblleft On nonlinear fractional programming, \textquotedblright \textit{ Management Science,} vol. 13, pp. 492--498, Mar. 1967. Available: [Online] http://www.jstor.org/stable/2627691.

\bibitem{bisection}
S. Boyd and L. Vandenberghe, \textit{ Convex optimization}, 2nd ed. Cambridge University Press, Cambridge, England, UK, 2004.

\bibitem{dyn}
M. L. Puterman, \textit{ Markov decision processes: discrete stochastic dynamic programming}, 1st ed. John Wiley \& Sons, Inc., New York, NY, USA, 1994.

\bibitem{prac}
X. Zhou, R. Zhang, and C. K. Ho, \textquotedblleft Wireless information and power transfer: architecture design and rate-energy tradeoff,\textquotedblright \textit{ IEEE Trans. Wireless Commun.}, vol. 61, no. 11, pp. 4754--4767, Nov. 2013.

\bibitem{markov}
P. Sadeghi, R. A. Kennedy, P. B. Rapajic, and R. Shams, \textquotedblleft Finitestate Markov modeling of fading channels -- a survey of principles and applications,\textquotedblright \textit{ IEEE Signal Processing Magazine}, vol. 25, pp. 57--80, Sep. 2008.


\end{thebibliography}
\end{document}